\documentclass[11pt]{article}

\usepackage{graphicx}
\usepackage{mathrsfs}
\usepackage{appendix}
\usepackage{color}
\usepackage{physics}

\usepackage{fullpage}
\usepackage{amsmath}
\usepackage{amsthm}
\usepackage{amssymb}
\usepackage{authblk}
\usepackage{bbm}
\usepackage[colorlinks,linkcolor=blue,citecolor=blue,urlcolor=magenta,filecolor=cyan]{hyperref}
\usepackage[capitalize]{cleveref}

\usepackage[utf8]{inputenc}
\usepackage[T1]{fontenc}

\newtheorem{theorem}{Theorem}
\newtheorem{lemma}[theorem]{Lemma}
\newtheorem{definition}[theorem]{Definition}

\newtheorem{corollary}[theorem]{Corollary}

\newtheorem{problem}[theorem]{Problem}

\renewcommand{\i}{\ensuremath{\text{\normalfont I}}}
\newcommand{\ii}{\ensuremath{\text{\normalfont I\!I}}}
\newcommand{\iii}{\ensuremath{\text{\normalfont I\!I\!I}}}
\newcommand{\iiii}{\ensuremath{\text{\normalfont I\!V}}}

\let\dsp=\displaystyle

\def\eg{{e.g.}}
\def\ie{{i.e.}}

\def\CL{\mathcal{L}}

\def\CO{\mathcal{O}}
\def\polylog{\mathrm{polylog}}
\def\poly{\mathrm{poly}}

\newcommand{\co}[1]{{\mathcal{O}\left(#1\right) }}

\def\bbc{\mathbb{C}}
\def\bbr{\mathbb{R}}
\def\rmt{\mathrm{T}}

\def\sz{{(0)}}

\title{Succinct Description and Efficient Simulation of Non-Markovian Open Quantum Systems}

\author[*]{Xiantao Li}
\author[$\dag$]{Chunhao Wang}

\affil[*]{Department of Mathematics, Pennsylvania State University}
\affil[$\dag$]{Department of Computer Science and Engineering, Pennsylvania State University}
\affil[ ]{Email: \{xiantao.li,cwang\}@psu.edu}

\date{}
\begin{document}
\maketitle

\begin{abstract}
Non-Markovian open quantum systems  {represent} the most general dynamics when the quantum system is coupled with a bath environment. The quantum dynamics arising from many important applications are non-Markovian. Although for special cases, such as Hamiltonian evolution and Lindblad evolution, quantum simulation algorithms have been extensively studied, efficient quantum simulations for the dynamics of non-Markovian open quantum systems remain underexplored. The most immediate obstacle for studying such systems is the lack of a universal succinct description of their dynamics. In this work, we fulfill the gap of studying such dynamics by 1) providing a succinct representation of the dynamics of non-Markovian open quantum systems with quantifiable error, and 2) developing an efficient quantum algorithm for simulating such dynamics with cost $\CO(t\, \polylog(t/\epsilon))$ for evolution time $t$ and precision $\epsilon$. Our derivation of the succinct representation is based on stochastic Schr\"odinger equations, which could lead to new alternatives to deal with open quantum systems as well.
  
\end{abstract}

\section{Introduction}
\label{sec:intro}

As the size of many modern-day electronic devices is continuously being reduced, quantum mechanical properties start to become
dominant. Many novel designs have emerged, \eg, quantum wires, quantum dots, and 
molecular transistors, to take advantage of these properties. What is common in these applications  is that the quantum properties are vital. However, these quantum dynamics do not evolve in isolation. Known as \emph{open quantum systems} \cite{breuer2002theory}, they are always interacting with their environment, which due to the large dimension, can not be included explicitly in the computation.   Another challenge comes from the fact that the continuous interactions can give rise to non-Markovian behavior, for which standard descriptions break down. The implication of non-Markovian dynamics to the quantum properties has been analyzed through many model examples \cite{breuer_non-markovian_2016,puebla2019spin,de_vega_dynamics_2017,semina_simulation_2016,shiokawa2004qubit,alliluev2021dynamics,makarov2020non}.
More importantly, there are important experimental observations of non-Markovian dynamics \cite{groblacher2015observation,madsen2011observation}, and such behavior has a strong impact on the electronic properties of molecular devices.  Furthermore, it has  been discovered that  non-Markovianity can enhance quantum entanglement \cite{thorwart2009enhanced}. 

Generally speaking, the dynamics of a quantum system interacting with a bath environment can be described by the von Neumann equation\footnote{The von Neumann equation is a generalization of the Schr\"{o}dinger equation to the context of density matrices.}
\begin{align}\label{eq: lvn-intro}
    i \partial_t \rho = [H_\text{tot}, \rho], \quad \rho(0)= \rho_S(0) \otimes \rho_B,
\end{align}
where the Hamiltonian $H_\text{tot}$ that couples the system to a bath is expressed as,
\begin{equation}\label{eq: htot-intro}
    H_\text{tot} = H_S \otimes I_B + I_S \otimes H_B + \lambda \sum_{\alpha=1}^M S_\alpha \otimes B_\alpha.
\end{equation}
Here, $\lambda$ is typically referred to as the \emph{coupling parameter}, and the integer $M$ indicates the number of interaction terms. To consider the problem in the context of quantum simulations, we let $2^n$ be the dimension of the system (S), \ie, $H_S \in \bbc^{2^n\times 2^n}$ is acting on $n$ qubits. A widely considered example is the spin Boson model, where $H_S$ and $S_\alpha$ are  Pauli matrices.  {One notable example in this category}, which is also relevant to quantum computers,  is the dynamics of qubits coupled to a boson bath \cite{wang2013exact}. Other  {applications}  include the {Anderson}-{Holstein} model~\cite{cao_lindblad_2017} where $H_S$ consists of the kinetic and potential energy terms, and those from electronic transport problems, where the Hamiltonian is expressed in terms of molecular or atomic orbitals \cite{galperin2005current}. 

When $\lambda = 0$, the system is completely decoupled from the environment, and the dynamics of the system reduces to \emph{Hamiltonian evolution}, which can also occur for finite $\lambda$, but under more stringent assumptions on the operator $S_\alpha$ and the initial state \cite{lidar1998decoherence}. In recent decades, the problem of simulating Hamiltonian evolution has been extensively studied. A subset of notable works can be found in~\cite{Llo96,BCG14,BCC15,BCK15,BCC17,LC17,LC19,CGJ19}. This line of research has led to optimal Hamiltonian simulation algorithms~\cite{LC17,GSLW19}.

In general, the continuous interactions with the bath conspire to a host of interesting quantum dynamics.  Known as open quantum systems \cite{breuer2002theory}, such problems have led to a long-standing challenge in modeling the system dynamics without an explicit representation of the bath.  Due to such modeling difficulty, the progress of developing fast simulation algorithms is much slower for open quantum systems. One exception is \emph{Markovian open quantum systems}, which arises when $0 < \lambda \ll 1$, \emph{and} the dynamics of the bath occurs at a much faster rate than the system. Intuitively, if an open quantum system is Markovian, the bath Hamiltonian is fast enough to ``forget'' the disturbance caused by the system-bath interaction, and therefore information only flows from system to bath with no transfer of information back to the system. By computing the infinitesimal  generator, a complete description of the Markovian dynamics has been obtained \cite{lindblad1976generators}, which later is referred to as the Lindblad equation.    In 2011, Kliesch, Barthel, Gogolin,  Kastoryano, and Eisert~\cite{KBG11} gave the first quantum algorithm for simulating Markovian open quantum systems with cost $\CO(t^2/\epsilon)$ for evolution time $t$ and precision $\epsilon$. In 2017, Childs and Li~\cite{CL17}  {presented an algorithm that improves} the cost to $\CO(t^{1.5}/\sqrt{\epsilon})$, and Cleve and Wang~\cite{CW17} further improved  {the algorithm by reducing } the complexity to $\CO(t\,\polylog(t/\epsilon))$, which is nearly optimal.  { All these algorithms are designed for models with sparse Hamiltonian and jump operators. Another notable approach is by Schlimgen et al.~\cite{schlimgen2021quantum,schlimgen2022quantum}, where the coefficient matrices in the Kraus form   associated with the Lindblad evolution are assumed to be given inputs.} Overall, the difficulty with simulating Markovian open quantum systems lies in the observation that the advanced algorithmic techniques for Hamiltonian simulation cannot be directly applied to open quantum systems because of the presence of decoherence. In fact, Cleve and Wang~\cite{CW17} have shown that it is impossible to achieve linear dependence on evolution time by a direct reductionist approach.

In sharp contrast to quantum Markovian processes,
there is no universal form for the quantum master equation for non-Markovian dynamics. In the regime when the rate of the bath dynamics is comparable to the rate of the system dynamics, the system becomes non-Markovian. Loosely speaking, the non-Markovian dynamics can be interpreted as the backflow of information for the environment to the open quantum system, and it can be more precisely characterized using various metrics \cite{breuer2009measure,vasile2011quantifying,luo2012quantifying}.  For such systems, many approaches have been developed to describe such dynamics in a way that goes beyond the Lindblad framework \cite{chruscinski_non-markovian_2010,de_vega_dynamics_2017,diosi1997non,gaspard1999non,ishizaki_quantum_2005,kelly2016generalized,meier_non-markovian_1999,montoya2016approximate,montoya2017approximate,pfalzgraff_efficient_2019,pollock_non-markovian_2018,shao2004decoupling,shi_new_2003,strathearn_efficient_2018,strunz1996linear,strunz_open_1999,suess_hierarchy_2014,tanimura2006stochastic,li2021markovian,tanimura2006stochastic}. The earliest approach dates back to the projection formalism of Nakajima and Zwanzig \cite{nakajima1958quantum,zwanzig1960ensemble}.  At the level of density matrices, the non-Markovian property is reflected in a memory integral that involves the bath correlation function, which can be approximated by a linear combination of Lorentzian terms \cite{ritschel2014analytic}.
This representation enables an embedding procedure, where the memory integral can be replaced by the dynamics of additional density-matrices \cite{meier_non-markovian_1999,li2021markovian}.  
 Although a unified framework is not yet present, such a procedure embeds   the density matrix of the system in an extended  system of equations involving additional density matrices, for which the dynamics is Markovian. The coefficients in the extended dynamics are connected to the spectral properties of the bath. Such a quantum master equation is  often referred to as the generalized quantum master equation (GQME). 

\subsection{Main results}
In this work, we 1) provide a succinct GQME representation for the dynamics of non-Markovian quantum systems, and 2) develop an efficient quantum algorithm for simulating the dynamics of open quantum systems  {based on} this new representation.  {The derivation of the GQME is an important mathematical contribution that provides the groundwork for algorithmic development.} It is worth noting that we only consider open quantum systems with coupling parameter $\lambda \ll 1$. This is the scenario when we have a priori bound on the model error, whereas in the strong coupling regime, it is difficult to determine such an error in advance,  and results can only be trusted with a leap of faith.

\paragraph{Mathematical contribution} 
In modeling open quantum systems, an important property to retain is the positivity. Toward this end, we choose to work with one that can be derived from an unraveling approach \cite{breuer2002theory}. Namely, there exists an underlying stochastic Schr\"odinger equation (SSE) and the state-matrix is automatically positive semidefinite. 
In order to accurately incorporate the effect of the bath, let $K$ be the number of Lorentzian terms in representing the bath correlation function \cite{ritschel2014analytic}. We consider an embedding of the system state $\rho_S$ into an $\CO(n\log K)$-qubit system (with dimension $K2^n$). Let $\Gamma$ denote the (unnormalized) state density matrix of this larger space. We derive from the SSE the following quantum master equation to describe the dynamics of $\Gamma$:
\begin{align}
  \label{eq:qme-intro}
  \partial_t \Gamma = -i(H\Gamma - \Gamma H^{\dag}) + \sum_{k=1}^KV_k\Gamma V_k^{\dag},
\end{align}
where $H$ (in general non-Hermitian) is defined in \cref{eq:H-block} which involves $H_S$, $S_{\alpha}$, and the bath correlation function (assumed to be part of the input).  $V_k$ will be defined later in \cref{eq:Vk} which involves the bath correlation function. The initial state $\Gamma(0)$, as deduced from the SSE, is given by, 
\begin{align}
  \label{eq:gamma0-intro}
  \Gamma(0) := \ketbra{0}{0}\otimes \rho_S(0) + \sum_{k=1}^K(\ketbra{4k-2}{4k-2} +3 \ketbra{4k})  
  \otimes \rho_S(0).
\end{align}
The system state is embedded in $\Gamma$ in the sense that the upper-left block of $\Gamma(t)$ --- the unnormalized state of the larger system at time $t$ --- approximates the exact system state $\rho_S(t)$ at time $t$ with error up to $\co{\lambda^3}$ (see \cref{eq: thm-rhos}). This error is referred to as the \emph{model error} as it characterizes the accuracy of modelling the actual dynamics determined by \cref{eq: lvn-intro} as a succinct GQME. Our model error of $\co{\lambda^3}$ improves   the model error $\co{\lambda^3/\alpha^3}$ of the Lindblad equation with $\alpha \ll 1$ representing the time scale separation between the system and the bath,  which (surprisingly) had not been precisely characterized until recently \cite{cao_lindblad_2017}. 

One quick illustrative  example of the GQME is a two-qubit model coupled to a common bosonic bath \cite{wang2013exact}, where HEOM type of equations were derived. Specifically, the system Hamiltonian is written as $H_s= \frac{\omega_0}2 (\sigma^\i_z +  \sigma^\ii_z)$, where $\i$ and $ \ii$  label the two qubits and $\omega_0$ is the Zeeman energy. In addition, in the coupling term, $S= \sigma^\i_x +  \sigma^\ii_x$. Thus $M=1.$ Furthermore, the study in \cite{wang2013exact} considered the bath correlation function $C(t)= \frac{\lambda \gamma}2 e^{-(\gamma+i\omega_0) t}.$
In light of \cref{eq:correlation}, we have that $K=1$,  $\theta_1=\sqrt{ \frac{\lambda \gamma}2},$   $\ket{Q_k}=1,$ and $d_1= \omega_0 + i \gamma.$  From the derivation in \cref{eq: T}, we also have $V_1= \theta_1 S.$ In this case, the matrix $\Gamma$ is a $5\times 5$ block matrix with total dimension being $20.$ 

\paragraph{Input model and simulation problem} 
The computational problem we consider is to simulate the dynamics generated by \cref{eq:qme-intro}, which consists of an initial state preparation problem and a target state approximation problem. More formally, we formulate the problem as follows.
\begin{problem}[Simulating non-Markovian open quantum systems]
  \label{prob:simulation}
  Consider the dynamics defined by \cref{eq:qme-intro}. {Suppose we are given access to some efficient descriptions of the operators in \cref{eq:qme-intro}.} For any initial state $\rho_S(0)$ of the system, evolution time $t$, and precision parameter $\epsilon$, we need to
  \begin{enumerate}
    \item prepare the initial state $\Gamma(0)$ as in \cref{eq:gamma0-intro}, and
    \item produce a quantum state $\widetilde{\rho}_S(t)$ for the system so that the trace-distance between this state and the upper-left block of $\Gamma(t)$, which is $\rho_S(t)$, is at most $\epsilon$, where $\Gamma(t)$ is the (unnormalized) state of evolving \cref{eq:qme-intro} for time $t$ with initial state $\Gamma(0)$.
  \end{enumerate}
\end{problem}

To solve this simulation problem, we need \emph{efficient descriptions} of the operators in \cref{eq:qme-intro}. The most straightforward input model is to assume that we are given these efficient descriptions directly. Such assumptions have been made in the literature of simulating Markovian open quantum systems~\cite{KBG11,CL17,CW17}. However, this straightforward input model is often not physically feasible: In many cases, we only have low-level information about the system, such as the system Hamiltonian and the system part of interaction Hamiltonians  {in \cref{eq: htot-intro}}, while high-level information such as descriptions of the operators in \cref{eq:qme-intro} is not readily available. With this practical consideration in mind,  we work with a low-level input model, \ie, we only assume information that arises in \cref{eq: htot-intro}, which is more convenient in real-world applications. 

In particular, for the operators $H_S$ and $S_{\alpha}$, we use a widely-used input model that has been recently introduced by Low and Chuang~\cite{LC19}, and Chakraborty, Gily{\'e}n, and Jeffery~\cite{CGJ19} --- the block-encodings of Hamiltonians. Roughly speaking, a block-encoding with normalizing factor $\alpha$ of a matrix $A$ is a unitary $U$ whose upper-left block is $A/\alpha$. This input model is general enough to include almost all efficient representations that arise in physics applications, including linear combinations of tensor products of Paulis, sparse-access oracles, and local Hamiltonians. {More specifically, if a matrix is $k$-local, then it is $2^k$-sparse. If a matrix $H$ is $d$-sparse, then it can be approximated as a linear combination of unitaries with sum-of-coefficients $\CO(d^2\norm{H}_{\max})$, where $\norm{H}_{\max}$ is the largest entry of $H$ in absolute value. Moreover, if a matrix can be written as a linear combination of matrices with the sum-of-coefficients $s$, then one can efficiently construct its block-encoding with normalizing constant $s$. Note that the inverse directions of the above implications are in general not true. Therefore, our algorithm also works when a $H_S$ and $S_{\alpha}$ are provided in a less general input model, such as a sum of local operators, sparse-access oracles, or linear combinations of unitaries.}

 {The BCF provides valuable information on how the bath influences the dynamics of the quantum system.} In this work, we consider a typical representation of the bath correlation function (BCF), expressed as,
\begin{align}
  \label{eq:correlation}
  {C}_{\alpha,\beta}(t) := \sum_{k=1}^K\theta_k^2\braket{\alpha}{Q_k}\braket{Q_k}{\beta}\exp(-i d_k^*t),
\end{align}
where $\ket{Q_k} \in \bbc^{M}$, $\theta_k \in \bbr$, and $d_k \in \bbc$. Note that for all $t$, ${C}_{\alpha, \beta}$ is an $M\times M$ matrix --- a size that is tractable for classical computers. The treatment of the BCF is at the heart of modeling open quantum systems, and in practice, the specification of the BCF starts with the spectral density (SD) $J_{\alpha,\beta}(\omega)$  that depends on the bath spectrum and interaction strength. For bosonic environment, they are related as follows,
\begin{equation}
    C_{\alpha,\beta } (t) =  \frac{1}{\pi}\int_{-\infty}^\infty \big(\coth(\frac{\omega}{2k_B T} ) \cos (\omega t) - i \sin (\omega t) \big) J_{\alpha,\beta} (\omega) d\omega, 
\end{equation} {
\eg, see \cite{lambert2020bofin}, and \cite{jin2008exact} for more general cases. }
Depending on the application, e.g., solvent, biomolecules, and nano materials, the spectral density can be adjusted accordingly. The pole expansion approach using Cauchy's Residue Theorem has been applied to  various types of SD functions \cite{ritschel2014analytic}.  In particular, both the functions $\coth(\frac{\omega}{2k_B T} )$ and $J(\omega)$ can be treated this way. Such an expansion, after a truncation beyond a cut-off frequency $d_{\max}$ for the poles, gives rise to a finite sum of complex exponential terms \cite{ritschel2014analytic}. The form of the coefficients in \cref{eq:correlation} is to ensure the positive definite property of the BCF. Namely, after the Fourier transform, the BCF has to be a semi positive definite function.   

Our algorithm will take such approximation results as an input. 
We consider the poles $d_k$ within a cut-off frequency $d_{\max}$, \ie, $\abs{d_k} \leq d_{\max}$.

 {With the weak coupling condition, and the explicit representation \cref{eq:correlation} of the BCF, we will derive a generalized quantum master equation as in \cref{eq:qme-intro}, with the extended quantum state $\Gamma(t)$ encoding the system density matrix $\rho_S(t).$ In addition, we demonstrate how the new Hamiltonian operator $H$  in \cref{eq:qme-intro}, as well as the jump operators $V_k$,  can be obtained from the coefficients in the BCF \cref{eq:correlation}. }

In summary, in our quantum algorithm, we assume we are given the following quantities, as an efficient description of \cref{eq:qme-intro}:
\begin{enumerate}
  \item a block-encoding $U_{H_S}$ of $H_S$;
  \item a block-encoding $U_{S_{\alpha}}$ of $S_{\alpha}$ for each $\alpha \in [M]$;
  \item the real numbers $\theta_k$, vectors $\ket{Q_k}$ (all entries), and complex numbers $d_k$ for $k \in [K]$ as in \cref{eq:correlation}, together with an upper bound $d_{\max}$ on $\abs{d_k}$.
\end{enumerate}

\paragraph{Algorithmic contribution} 
Our main algorithmic contribution is a quantum algorithm that solves \cref{prob:simulation}. We informally state this result as follows.
\begin{theorem}[Informal version of \cref{thm:qalg}]
  Suppose that we are given a block-encoding $U_{H_S}$ of $H_S$, block-encodings $U_{S_{\alpha}}$ of $S_{\alpha}$ (for $\alpha \in [M]$), $\theta_k$, $\ket{Q_k}$ (all entries),   {$\lambda$,} and $d_k$ for $k \in [K]$ as in \cref{eq:correlation}. Then there exists a quantum algorithm that solves \cref{prob:simulation} using 
  \begin{align}
    \co{t\,\polylog(t/\epsilon) \poly( {\mu}, M, K, \lambda, d_{\max})},
  \end{align}
  queries to $U_{H_S}$, and $U_{S_a}$ and additional 1- and 2-qubit gates,
  where  {$\mu$} is the normalizing factor of the block-encodings, $M$ is the number of interaction terms in \cref{eq: lvn-intro}, $K$ is the number of Lorentzian terms in \cref{eq:correlation}, and  $d_{\max}$ is an upper bound of $\abs{d_k}$. 
\end{theorem}

Our algorithm follows the high-level idea of~\cite{CW17}, but we have generalized their techniques to work with block-encoded inputs. The building block of our algorithm is an implementation of completely positive maps given block-encodings of the Kraus operators (see \cref{lemma:block-encoding-channel} for more details). Suppose the normalizing factors of the block-encodings of the Kraus operators are $\alpha_1, \ldots, \alpha_m$, then our construction gives the success probability parameter $1/\sum_{j=1}^m\alpha_j^2$, while a straightforward construction using Stinespring dilation yields a worse success probability parameter $1/(\sum_{j=1}^m\alpha_j)^2$, which does not permit the desired dependence on $t$ and $\epsilon$. 

Since the dynamics that \cref{eq:qme-intro} generates is completely positive, we consider an \emph{infinitesimal approximation map} that approximates   {the first-order Taylor approximation of} the dynamics for a small enough evolution time. From the low-level input model, we can construct the block-encodings of the Kraus operators of this infinitesimal approximation map, and it can be implemented using our building block \cref{lemma:block-encoding-channel} with high success probability parameter. We repeat this construction until the success probability parameter becomes a constant, and at that point, we have obtained a normalized version of $\Gamma(t)$ for $t$ proportional to a constant. Now, to extract the upper-left block of the resulting (normalized) density matrix, we use oblivious amplitude amplification for isometries (\cite{CW17} and \cref{lemma:oaa}) to achieve this with an extra factor $\sqrt{K}$, as the trace of $\Gamma(t)$ is upper bounded by $\CO(K)$ for all $t$ (see \cref{eq: tr-bound}). 

  {For the problem of simulating Hamiltonian dynamics, high-order Taylor approximation yields simpler and faster quantum algorithms, e.g.~\cite{BCC15}. However, for simulating open quantum systems, it is not known how to take advantage of higher-order approximations. This is because high powers of the superoperator defined in \cref{eq:qme-intro} are too complicated to keep track of its completely positive structure, which is the key to implementing such maps.}

\subsection{Summary of contributions}
We highlight our contributions as follows:
\begin{enumerate}
  \item We provide a succinct representation of non-Markovian dynamics, where the density matrix from the GQME is consistent with that from the full quantum dynamics with provable  $\co{\lambda^3}$ accuracy. A notable feature of our new succinct representation is that the positivity is guaranteed, which is the key requirement for designing quantum simulation algorithms.

  \item We develop a quantum algorithm based on this representation. The cost of our algorithm scales linearly in $t$, poly-logarithmically in $\epsilon$, and polynomially in $M$ and $K$. To the best of our knowledge, this algorithm is the first to achieve linear dependence on $t$ and poly-logarithmic dependence on $\epsilon$ for simulating non-Markovian open quantum systems. In addition, our algorithm works with low-level input models, which are readily available in many real-world applications.
  
  \item Other technical contributions: We have shown that the GQMEs can be unraveled (see \cite{breuer2002theory}) into stochastic Schr\"odinger equations, which provides another potential alternative to obtain the density matrix  {as a statistical quantity}. In addition, we prove that the extended density matrix from the GQME is bounded over the time scale $\co{\lambda^{-1}}.$   
\end{enumerate}

\subsection{Related work}
In the context of modeling open quantum systems \cite{breuer2002theory}, the hierarchical equations of motion (HEOM) approach \cite{tanimura2006stochastic,tanimura2020numerically} also yields an extended dynamics for the density matrix, but without using the weak coupling assumption.
Rather, the equations are truncated based on numerical observations. In principle, the quantum algorithms in this work can also be applied to those GQMEs from the HEOM approach,  { provided that the approach can be proved to produce completely positive maps}.  Meanwhile, since no error bound is available, it is difficult to prescribe the level of truncation in advance, which also makes it difficult to estimate the computational complexity.

Another interesting development is to embody the memory effect using a local form of the GQME, where the generator consists of multiple Lindblad operators with time-dependent coefficients. In fact, it has been proved \cite{chruscinski_non-markovian_2010} that any nonlocal form of  the non-Markovian dynamics can be rewritten in a local form with \emph{time-dependent} generators.  {Due to the fact that the proof is non-constructive,  the implementations of this type of Lindblad operators are empirical. }
Sweke, Sanz, Sinayskiy, Petruccione, and Solano~\cite{Swe+16} considered such time-local quantum master equations, and constructed quantum algorithms.  
Their algorithms rely on Trotter splittings, by decomposing the entire generator into local operators. Since it is not yet clear how the GQMEs from the current approach, or those from the HEOM approach, can be expressed in time-local forms, a direct comparison of the computational complexity is not yet available. 

\vspace{1em}
\subsection{Open questions}
Modeling open quantum systems outside the weak coupling regime is still an outstanding challenge. The HEOM approach \cite{tanimura2020numerically} relies on a frequency cut-off to achieve a Markovian embedding. Quantifying the error associated with such an approximation, and ensuring the positivity are two of the remaining issues.   

Another interesting scenario is when the open quantum system is subject to an external potential. Quantum optimal control is one important example. Deriving a GQME in the presence of a time-dependent external field while still maintaining the control properties is still an open problem to the best of our knowledge.  

The cost of our quantum algorithm is $\CO(t\,\polylog(t/\epsilon))$. Is there a faster quantum algorithm that achieves an additive cost, \ie, $\CO(t + \polylog(1/\epsilon))$? This additive cost is the lower bound for Hamiltonian simulation~\cite{BACS07,BCK15,BCC17}, and it is hence a lower bound for the non-Markovian simulation problem. The optimal Hamiltonian simulation was achieved by quantum signal processing due to Low and Chuang~\cite{LC17}, which has been generalized to quantum singular value transformation by Gily{\'e}n, Su, Low, and Wiebe~\cite{GSLW19}. Unfortunately, these techniques do not immediately generalize to open quantum systems, as it is not clear what the correspondence of singular values and eigenvalues should be for superoperators.

\section{Preliminaries}
\label{sec:prelim}

\subsection{Notation}
\label{sec:notation}

In this paper, we use the ket-notation to denote a vector only when it is normalized. For a vector $v$, we use $\norm{v}$ to denote its \emph{Euclidean norm}. For a square matrix $M$, we use $\norm{M}$ to denote its \emph{spectral norm} and use $\norm{M}_1$ to denote its \emph{trace norm}, \ie, $\norm{M}_1 = \tr(\sqrt{M^{\dag}M})$. The identity operator acting on a Hilbert space of dimension $N$ is denoted by $I_N$, \eg, $I_{2^n}$ is the identity operator acting on $n$ qubits. When the context is clear, we drop the subscript and simply use $I$. We use calligraphic fonts, such as $\mathcal{K}$, $\mathcal{L}$, and $\mathcal{M}$ to denote \emph{superoperators}, which maps matrices to matrices. We consider superoperators of the form
\begin{align}
  \mathcal{M}: \bbc^{N \times N} \rightarrow \bbc^{M \times M},
\end{align}
and use $\rmt(\bbc^N, \bbc^M)$ to denote the set of all such superoperators. In particular, we use $\mathcal{I}$ to denote the \emph{identity map}, which maps every matrix to itself. Whenever necessary, we use the subscript $N$ of $\mathcal{I}_N \in \rmt(\bbc^N, \bbc^N)$ to highlight the dimension of matrices it acts on. For example, $\mathcal{I}_{2^n}$ is acting on $n$-qubit operators. The \emph{induced trace norm} of a superoperator $\mathcal{M} \in \rmt(\bbc^N, \bbc^M)$, denoted by $\norm{\mathcal{M}}_1$, is defined as
\begin{align}
  \norm{\mathcal{M}}_1 = \max\{\norm{\mathcal{M}(A)}_1: A \in \bbc^{N \times N}, \norm{A}_1 \leq 1\}.
\end{align}
The \emph{diamond norm} of a superoperator $\mathcal{M} \in \rmt(\bbc^N, \bbc^M)$, denoted by $\norm{\mathcal{M}}_{\diamond}$ is defined as
\begin{align}
  \norm{\mathcal{M}}_{\diamond} := \norm{\mathcal{M}\otimes\mathcal{I}_N}_1.
\end{align}

For a positive integer $m$, we use $[m]$ to denote the set $\{1, 2, \ldots, m\}$.

\subsection{The block-encoding method}
We use the notion of \emph{block-encoding} as an efficient description of input operators. To have access to an operator $A$, we assume we have access to a unitary $U$ whose upper-left block encodes $A$ in the sense that
\begin{align}
  U = 
  \begin{pmatrix}
    A/\alpha & \cdot\\
    \cdot & \cdot\\
  \end{pmatrix},
\end{align}
which implies that $A = \alpha(\bra{0}\otimes I)U\ket{0}\otimes I)$. More precisely, we have the following definition.

\begin{definition}
  Let $A$ be an $n$-qubit operator. For a positive real number $\alpha > 0$ and natural number $m$, we say that an $(n + m)$-qubit unitary $U$ is an $(\alpha, m, \epsilon)$-\emph{block-encoding} of $A$ if
  \begin{align}
    \norm{A - \alpha(\bra{0}\otimes I_{2^n})U\ket{0}\otimes I_{2^n})} \leq \epsilon.
  \end{align}
\end{definition}

The following lemma shows how to construct a block-encoding for sparse matrices.
\begin{lemma}[{\cite[Lemma 48]{GSLW19}}]
  \label{lemma:sparse-to-be}
  Let $A\in\bbc^{2^n\times 2^n}$ be an $n$-qubit operator with at most $s$ nonzero entries in each row and column. Suppose $A$ is specified by the following sparse-access oracles:
  \begin{align}
    \label{eq:sparse-1}
    O_A: &\ket{i}\ket{j}\ket{0} \mapsto \ket{i}\ket{j}\ket{A(i, j)}, \text{ and}\\
    \label{eq:sparse-2}
    O_S: &\ket{i}\ket{k} \mapsto \ket{i}\ket{r_{i,k}},
  \end{align}
  where $r_{i,k}$ is the $k$-th nonzero entry of the $i$-th row of $A$. Suppose $\abs{A_{i,j}} \leq 1$ for $i \in [m]$ and $j\in[n]$. Then for all $\epsilon \in (0, 1)$, an $(s, n+3, \epsilon)$-block-encoding of $A$ can be implemented using $\CO(1)$ queries to $O_A$ and $O_S$, along with $\CO(n+\polylog(1/\epsilon))$ 1- and 2-qubit gates.   {Moreover, if $A_{i,j} \in \{0, 1\}$ for all $i \in [m]$ and $j \in [n]$, the block-encoding of $A$ can be implemented precisely, i.e., $\epsilon = 0$.\footnote{The case when $A_{i,j} \in \{0, 1\}$ was not explicitly stated in \cite[Lemma 48]{GSLW19}; however, the conclusion is not hard to obtain as a special case of their proof.}}
\end{lemma}

A linear combination of block-encodings can be constructed using \cite[Lemma 52]{GSLW19}. Here, we slightly generalize their construction to achieve better performance when the normalizing factors of each block-encoding are different.

\begin{lemma}
  \label{lemma:sum-to-be}
  Suppose $A := \sum_{j=1}^m y_j A_j \in \bbc^{2^n\times 2^n}$, where $A_j \in \bbc^{2^n \times 2^n}$ and $y_j > 0$ for all $j \in \{1, \ldots m\}$. Let $U_j$ be an $(\alpha_j, a, \epsilon)$-block-encoding of $A_j$, and $B$ be a unitary acting on $b$ qubits (with $m \leq 2^b-1$) such that $B\ket{0} = \sum_{j=0}^{2^b-1}\sqrt{\alpha_jy_j/s}\ket{j}$, where $s = \sum_{j=1}^my_j\alpha_j$. Then a $(\sum_{j}y_j\alpha_j, a+b, \sum_{j}y_j\alpha_j\epsilon)$-block-encoding of $\sum_{j=1}^my_jA_j$ can be implemented with a single use of $\sum_{j=0}^{m-1}\ketbra{j}{j}\otimes U_j + ((I - \sum_{j=0}^{m-1}\ketbra{j}{j})\otimes I _{\bbc^{2^a}}\otimes I_{\bbc^{2^{n}}})$ plus twice the cost for implementing $B$.

\end{lemma}

\begin{proof}
  The proof is similar to that of \cite[Lemma 52]{GSLW19}. The difference is that, instead of preparing the state $\sum_{j=1}^m \frac{\sqrt{y_j}}{\sqrt{\sum_{j=1}^my_j}}\ket{j}$, we use the state preparation gate $B$ here.
  First note that
  \begin{align}
    \norm{(\bra{0}^{\otimes a}\otimes I_{2^n})U_j(\ket{0}^{\otimes a}\otimes I_{2^n}) - A_j/\alpha_j} \leq \epsilon.
  \end{align}
  Let $W = \sum_{j=0}^{m-1}\ketbra{j}{j}\otimes U_j + ((I - \sum_{j=0}^{m-1}\ketbra{j}{j})\otimes I _{2^a}\otimes I_{2^{n}})$ and define $\widetilde{W} = (B^{\dag}\otimes I_{2^a} \otimes I_{2^n}) W (B\otimes I_{2^a} \otimes I_{2^n})$. We have
  \begin{align}
    &\norm{\sum_{j=1}^my_jA_j - s\left(\bra{0}^{\otimes b}\otimes \bra{0}^{\otimes a} \otimes I_{2^n}\right)\widetilde{W}\left(\ket{0}^{\otimes b}\otimes \ket{0}^{\otimes a} \otimes I_{2^n}\right)} \\
    &=\norm{\sum_{j=1}^my_jA_j - \sum_{j=1}^m\alpha_jy_j\left(\bra{0}^{\otimes a}\otimes I_{2^n})U_j(\ket{0}^{\otimes a}\otimes I_{2^n})\right)} \\
    &\leq \sum_{j=1}^my_j\alpha_j \norm{A_j/\alpha_j - \left(\bra{0}^{\otimes a}\otimes I_{2^n})U_j(\ket{0}^{\otimes a}\otimes I_{2^n})\right)} \\
        &\leq \sum_{j=1}^my_j\alpha_j\epsilon.
  \end{align}
\end{proof}

\subsection{Technical tool for implementing completely positive maps}
\label{sec:technicaltools}
In this subsection, we provide the technical primitives for developing a simulation algorithm. The following lemma generalizes the technique of linear combination of unitaries for completely positive maps~\cite{CW17} to the context of block-encoding. This tool might be of independent interest as well.

\begin{lemma}
  \label{lemma:block-encoding-channel}
  Let $A_1, \ldots, A_{m} \in \bbc^{2^n}$ be the  Kraus operators  {of a completely positive map $\mathcal{M}$ \cite{kraus1971general}, 
  \[\mathcal{M} (\rho) = \sum_{j=1}^m A_j \rho A_j^\dagger.\]
  }  Let $U_1, \ldots, U_m \in \bbc^{2^{n+n'}}$ be their corresponding $(s_j, n', \epsilon)$-block-encodings, \ie, 
  \begin{align}
    \norm{A_j - s_j (\bra{0}\otimes I)U_j\ket{0}\otimes I)} \leq \epsilon, \quad \text{ for all $1 \leq j \leq m$}.
  \end{align}
  Let $\ket{\mu}:= \frac{1}{\sqrt{\sum_{j=1}^{m}s_j^2}}\sum_{j=1}^{m}s_j\ket{j}$. Then $(\sum_{j=1}^m\ketbra{j}{j}\otimes U_j)\ket{\mu}\ket{0}\otimes I$ implements this completely positive map  { $\mathcal{M}$ } in the sense that
  \begin{align}\label{eq: M-apply}
    \norm{I\otimes \bra{0}\otimes I \left(\sum_{j=1}^m\ketbra{j}{j}\otimes U_j\right) \ket{\mu}\ket{0}\ket{\psi} - \frac{1}{\sqrt{\sum_{j=1}^ms_j^2}}\sum_{j=1}^m \ket{j}A_j\ket{\psi}} \leq \frac{m\epsilon}{\sqrt{\sum_{j=1}^ms_j^2}}
  \end{align}
  for all $\ket{\psi}$.
\end{lemma}

\begin{proof}
  It is easy to verify that
  \begin{align}
    \left(\sum_{j=1}^m\ketbra{j}{j}\otimes U_j\right)\ket{\mu}\ket{0}\ket{\psi} = \frac{1}{\sqrt{\sum_{j=1}^ms_j^2}}\sum_js_j\ket{j}U_j(\ket{0}\ket{\psi}).
  \end{align}
 {  With a direct substitution into \cref{eq: M-apply}, one arrives at,}
  \begin{align}
    &\norm{(I\otimes \bra{0}\otimes I) \left(\sum_{j=1}^m\ketbra{j}{j}\otimes U_j\right) \ket{\mu}\ket{0}\ket{\psi} - \frac{1}{\sqrt{\sum_{j=1}^ms_j^2}}\sum_j\ket{j} A_j\ket{\psi}} \\
    &= \norm{(I\otimes \bra{0}\otimes I) \frac{1}{\sqrt{\sum_{j=1}^ms_j^2}}\sum_js_j\ket{j}U_j(\ket{0}\ket{\psi}) - \frac{1}{\sqrt{\sum_{j=1}^ms_j^2}}\sum_j\ket{j} A_j\ket{\psi}}\\
    &= \frac{1}{\sqrt{\sum_{j=1}^ms_j^2}}\norm{\sum_js_j\ket{j} (\bra{0}\otimes I)U_j(\ket{0}\otimes I)\ket{\psi} - \sum_j\ket{j} A_j\ket{\psi}} \\
&\leq \frac{m\epsilon}{\sqrt{\sum_{j=1}^ms_j^2}}.
  \end{align}
\end{proof}

\section{Non-Markovian Quantum Master Equation}
In this section, we present the derivation of the GQME \cite{li2021markovian}. Non-Markovian dynamics has been extensively studied in the context of open quantum systems \cite{breuer2002theory}.
The starting point for considering an open quantum system is a quantum dynamics that couples a quantum system and a bath environment,
\begin{equation}\label{eq: lvn}
    i \partial_t \rho = [H_\text{tot}, \rho], \quad \rho(0)= \rho_S(0) \otimes \rho_B,
\end{equation}
where the coupled Hamiltonian is given by,
\begin{equation}\label{eq: htot}
    H_\text{tot} = H_S \otimes I_B + I_S \otimes H_B + \lambda \sum_{\alpha=1}^M S_\alpha \otimes B_\alpha.
\end{equation}
Here $H_S \in \bbc^{2^n\times 2^n}$ is acting on $n$ qubits, and $M$ refers to the number of interaction terms.

We follow the standard setup by choosing $\rho_B$ according to a statistical ensemble, \eg, 
\begin{equation}
    \rho_B = \frac{\exp (-\frac{H_B}{k_BT})}Z, \quad Z:= \tr \left(\exp (-\frac{H_B}{k_B T})\right),
\end{equation}
 { with $k_B$ and $T$ being respectively the Boltzmann constant and the temperature.  As a result, $[H_B, \rho_B]=0.$}  Without loss of generality, we can assume that $\text{tr}(\rho_B B_\alpha)=0, \; \alpha=1,2,\cdots,M,$ which can be ensured by properly shifting of the operators \cite{gaspard1999non}. This helps to eliminate $\mathcal{O}(\lambda)$ terms in the asymptotic expansion \cite{gaspard1999non}.

 {
Of primary interest in the theory of open quantum systems is the density matrix of the system,
$\rho_S(t)$, which in principle can be obtained with a partial trace: $\rho_S(t)=\text{tr}_B\left(\rho(t)\right).$  } To arrive at a  quantum master equation that embodies non-Markovian properties, we start with the non-Markovian stochastic Schr\"odinger equation (NMSSE), which was derived from  the wave function representation of \cref{eq: lvn}  in \cite{gaspard1999non}, and later revisited in \cite{biele2012stochastic}. In NMSSE, a stochastic realization of the quantum state follows a stochastic differential equation,
\begin{equation}\label{eq: sse}
    i \partial_t \psi = \hat H_S \psi - i \lambda^2\sum_{\alpha,\beta=1}^M  \int_0^t C_{\alpha,\beta}(\tau) S^\dagger_\alpha  
    e^{-i \hat H_S \tau } S_\beta   \psi(t-\tau) d\tau +\lambda \sum_{\beta=1}^M   \eta_\beta(t)  S_\beta   \psi(t).
\end{equation}
Here $i=\sqrt{-1}$. The matrix $C(t): \mathbb{R} \to \mathbb{C}^{M\times M}$ with elements $C_{\alpha,\beta}(t)$ corresponding to the correlation among the bath correlation $\{B_\alpha\}_{1\le \alpha \le M}$.

Each noise term $\eta_\alpha(t): \mathbb{R} \to \mathbb{C}$ is Gaussian with mean zero and correlation given by \cite{biele2012stochastic},
\begin{equation}\label{eq: fdt}
    \mathbb{E}[ \eta_\alpha^*(t)\eta_\beta(t') ]= C_{\alpha,\beta}(t-t'), \quad 1\leq \alpha, \beta \leq M.
\end{equation}
The stationarity of the process also implied that $C(t)=C(-t)^\dagger$. Thus, it suffices to consider the correlation function for $t\geq 0.$ This relation between a dissipation kernel and the time correlation of the noise is well-known in non-equilibrium statistical physics,
and it is often labeled as the second fluctuation-dissipation theorem \cite{kubo1966fluctuation}.

\subsection{The generalized quantum master equation (GQME)}\label{eq: gqme-derive}

With a scale separation assumption, the NMSSE can be reduced to the Lindblad equation \cite{cao_lindblad_2017,lidar2001completely}.  {In this regime, the bath correlation behaves like a delta function \cite{gaspard1999non},
\begin{equation}\label{eq: markA}
     C(t-t') \approx \sum_{j=1}^M  \theta_j^2 \ketbra{R_j}{R_j} \delta(t-t'),
\end{equation}
which  simplifies the memory integrals to local terms. By defining operators $L_j$,
\begin{equation}\label{eq: V}
  { L_{j}}=\sum_{\beta=1}^M  \theta_j \braket{R_{j}}{\beta} S_\beta,
\end{equation} 
the NMSSE then implies the following Lindblad equation \cite{biele2012stochastic},
\begin{equation}
\label{eq: lindblad}
i \partial_t \rho_S = [H_S, \rho_S] - \lambda^2 {\sum_{j=1}^M \Big(  L^\dagger_j 
    L_j \rho_S + \rho_S  L^\dagger_j 
    L_j - 2  L_j \rho_S
    L^\dagger_j \Big).}
\end{equation}}

However, in the non-Markovian regime, the assumption \cref{eq: markA} breaks down,  {i.e., the correlation length is finite}, and a different approach is needed to derive a quantum master equation that governs the dynamics of $\rho_S$. In this paper, we follow the general unraveling approach \cite{breuer2002theory}.
In \cref{a: deriv} we show that the memory effect can be embedded in an extended stochastic system by introducing auxiliary wave functions, $\chi^\i$ to $\chi^\iiii.$ The dynamics can be summarized in the following form, 
\begin{equation}\label{eq: ext-iv}
\left\{
    \begin{array}{l}
        i \partial_t \psi =\dsp  \hat H_S \psi - i \lambda    \sum_{k=1}^K T_k^\dagger  \chi^\i_k - i \lambda   \sum_k T_k \chi^\ii_k, \\
    \begin{array}{ll}
            i\partial_t \chi^\i_k =& (\hat H_S + d^*_k ) \chi^\i_k + i \lambda T_k \psi(t),\\
       i\partial_t \chi^\ii_k = & (\hat H_S - d_k) \chi^\ii_k  -i \lambda  T_k^\dagger \chi^\iii_k -i \lambda  T_k \chi^\iiii_k +  i  \gamma_k \psi(t) \dot{w}_k(t),\\
       i\partial_t \chi^\iii_k = & i \lambda T_k \chi_k^\ii +  (\hat H_S - d_k-d_k^*) \chi^\iii_k + i\gamma_k \chi_k^\i \dot{w}_k(t),\\
       i\partial_t \chi^\iiii_k = & i \lambda T_k^\dagger \chi_k^\ii +  (\hat H_S - 2d_k) \chi^\iiii_k + 2 i\gamma_k \chi_k^\ii \dot{w}_k(t),\\
    \end{array}
        \end{array}
     \right.
\end{equation}
for $  k=1,2,\ldots, K.$  {Due to the presence of the auxiliary wave functions, the dynamics of $\psi$ is non-Markovian, and the additional equations induce a memory effect that imitates the non-Markovian behavior.}
We have dropped  $\mathcal{O}(\lambda)$ terms in the last two equations,  { which can be justified as follows, by a substitution into the third equation, one can see that this truncation will contribute to an  $\mathcal{O}(\lambda^2)$ error to the dynamics of $\chi^\ii,$ which, after another substitution, leads to an  $\mathcal{O}(\lambda^3)$ perturbation
in the first equation in \cref{eq: ext-iv}.} Namely, the accuracy is the same as the NMSSE \cref{eq: sse}.  The operators $T_k$ are defined  {in terms of the operators $S_\beta$ in \cref{eq: htot-intro} and the coefficients in the BCF \cref{eq:correlation} as follows,}
\begin{equation}\label{eq: T} 
  T_{k}=\sum_{\beta=1}^M  \theta_k \braket{Q_{k}}{\beta} S_\beta,
\end{equation} 

From the definitions of the auxiliary functions, we can deduce their initial conditions,
\[ \chi_k^\i(0)=0, \quad  \chi_k^\ii(0)= i \zeta_k(0) \psi(0),\quad \chi_k^\iii(0)=0,\quad \chi_k^\iiii(0)= -\zeta_k(0)^2 \psi(0),\]{
with $\zeta_k$'s being independent Gaussian random variables of zero mean and unit variance.}
 
 To derive the corresponding GQME 
of the stochastic model in \cref{eq: ext-ii}, we first write it  as a system of SDEs,
\begin{equation}\label{eq: Psi-sde}
  i \partial_t \Psi = H \Psi + \sum_{k=1}^K V_k \Psi  \dot{w}_k.
\end{equation}
Here the function $\Psi$ includes the wave function $\psi$ and all the auxiliary wave functions $\{\chi_k^\i, \chi_k^\ii, \chi_k^\iii, \chi_k^\iiii\}_{k=1}^K$. One can arrange the wave functions and the operator $V_1$ as follows,
\begin{equation}\label{eq: Psi-V_k}
  \Psi = \left(
  \begin{array}{c}
    \psi \\
    \chi_1^\i\\
    \chi_1^\ii\\
        \chi_1^\iii\\
    \chi_1^\iiii\\
    \chi_2^\i\\
    \chi_2^\ii\\
        \chi_2^\iii\\
    \chi_2^\iiii\\
   \vdots 
  \end{array}
  \right),\;
  V_1 =\left(\begin{array}{ccccccc} 
  0 & 0& 0 & 0 & 0 & 0 &\dots  \\ 
  0 & 0 & 0 & 0 & 0 & 0 &\dots \\
  \sqrt{2 \nu_1}   & 0& 0 & 0 & 0 & 0 &\dots \\ 
0&  \sqrt{2 \nu_1}   & 0& 0 & 0 & 0 &\dots \\ 
0& 0 & 2\sqrt{2 \nu_1}   & 0& 0  & 0 &\dots \\ 
  0 & 0& 0 & 0 & 0  & 0 &\dots \\ 
  0 & 0 & 0 & 0 & 0 & 0 &\dots \\
  0 & 0& 0 & 0 & 0  & 0 &\dots \\ 
  0 & 0 & 0 & 0 & 0 & 0 &\dots \\
\vdots &\vdots &\vdots &\vdots &\vdots &\vdots &\ddots 
 \end{array}\right) \otimes I_S.  \; 
  \end{equation}
  Similarly, 
  \[ V_2 =\left(\begin{array}{ccccccccc} 
 0 & 0& 0 & 0& 0 & 0 & 0  & 0 &\dots \\ 
 0 & 0& 0 & 0 & 0 & 0 & 0  & 0 &\dots\\
0 & 0&    0 & 0& 0 & 0 & 0  & 0 &\dots \\ 
0 & 0&  0 & 0 & 0 & 0 & 0 & 0 &\dots \\
0 & 0&  0 & 0& 0 & 0 & 0 & 0 &\dots  \\ 
0 & 0&  0 & 0 & 0 & 0 & 0 & 0 &\dots \\
  \sqrt{2 \nu_2}   &0 & 0& 0& 0 & 0 & 0 & 0 &\dots \\ 
0 & 0&0& 0& 0 &  \sqrt{2 \nu_2}   & 0 & 0 &\dots \\ 
0 & 0&0& 0 & 0& 0 & 2\sqrt{2 \nu_2}   &  0 &\dots  \\ 
\vdots &\vdots& \vdots &\vdots &\vdots &\vdots &\vdots & \vdots &\ddots
 \end{array}\right) \otimes I_S.\]
Clearly, $V_k$'s are sparse and low rank. 
Consider the standard basis in $\mathbb{R}^{4K+1}$, here abbreviated simply into $\ket{0}, \ket{1}, \ldots, \ket{4K}$. A careful inspection reveals the following closed-form expressions,
\begin{align}
  \label{eq:Vk}
  V_k = \sqrt{2 \nu_k}  \Big( \ketbra{ {4k-2}}{0} +  \ketbra{ {4k-1}}{ {4k-3}} +  2\ketbra{{4k}}{ {4k-2}}\Big) \otimes I_S, 
\end{align}
for $k=1, 2, \ldots, K.$
Here $\nu_k=\text{Im} d_k,$ and they are nonnegative.

The operator $H$ in \cref{eq: Psi-sde} can also be written in a block form,
{\scriptsize
\begin{equation}\label{eq; Hnew}
\arraycolsep=1.1pt\def\arraystretch{1.3}
\left[\begin{array}{cccccccccc}
 H_s & -i \lambda  T_1^\dagger  & -i \lambda  T_1  &  0 & 0 &  -i \lambda T_2^\dagger &  -i \lambda  T_2 & 0& 0&   \cdots  \\
 i \lambda T_1 & H_s \!+\! d_1^* & 0 & 0 & 0 & 0& 0& 0& 0& \cdots \\
  0    &0  & H_s \!-\! d_1 &  -i \lambda T_1^\dagger &  -i \lambda  T_1 & 0& 0& 0&
   0 &\cdots \\
0 & 0 & i \lambda  T_1  &  H_s \! -\!2 i\nu_1  &  0 & 0& 0& 0& 0& \cdots \\
0 & 0 &i \lambda T_1^\dagger &0  &  H_s\! -\! 2d_1& 0& 0& 0& 0 & \cdots \\
i \lambda T_2 & 0& 0& 0& 0& H_s \!+\! d_2^* & 0 & 0 & 0  & \cdots \\
0& 0& 0&
   0 & 0 & 0  & H_s\! -\! d_2 &  -i \lambda T_2^\dagger &  -i \lambda  T_2  &\cdots \\
0& 0& 0&
   0 & 0 & 0 & i \lambda  T_2  &  H_s  \!-\!2 i\nu_2 & 0 & \cdots \\
0& 0& 0&
   0 & 0 & 0 & i \lambda  T_2  & 0 &  H_s  \!-\! 2d_2 & \cdots \\
\vdots & \vdots & \vdots  & \vdots & \vdots & \vdots & \vdots  & \vdots & \vdots & \ddots
\end{array}\right].
\end{equation}
}

The extended stochastic dynamics \cref{eq: ext-iv} introduced an auxiliary space that mimics the effect of the quantum bath.  
Recall 
$I_S \in \mathbb{C}^{2^n\times 2^n}$
is the identity operator. Similarly, we let 
$I_A \in \mathbb{C}^{(4K+1)\times(4K+1)}$ be the identify operator in an auxiliary space labelled by $A$. Then the Hamiltonian in \cref{eq; Hnew} can be expressed in terms of tensor products,
\begin{equation}
  \label{eq:H-block}
    H = I_A \otimes H_S + H_A \otimes I_S + i\lambda  \sum_{k=1}^K D_k \otimes T_k + i\lambda \sum_{k=1}^K E_k \otimes T_k^\dagger.
\end{equation}
Here 
$H_A \in \mathbb{C}^{(4K+1)\times(4K+1)}$ is a diagonal matrix: 
\[ H_A = \text{diag}\left\{0, d_1^*, -d_1, d_1^*-d_1, -2d_1, d_2^*, -d_2, d_2^*-d_2, -2d_2, \ldots   \right\}.\] 
In addition, the matrices $D_k, E_k \in \mathbb{C}^{(4K+1)\times (4K+1)}$ are given by
\begin{equation}\label{eq: DkEk}
\begin{aligned}
 D_k=& \ketbra{{4k-3}}{0} - \ketbra{0}{{4k-2}} 
 + \ketbra{{4k-1}}{{4k-2}} - \ketbra{{4k-2}}{{4k-1}},\\
 E_k=& - \ketbra{0}{{4k-3}} 
 + \ketbra{{4k}}{{4k-2}} - \ketbra{{4k-2}}{{4k}}.\\
\end{aligned}
\end{equation}
Assuming that the dimension of the original wave function is $n$, \ie, $\ket{\psi} \in \mathbb{C}^n$, the dimension of $\Psi$ is $(4K+1)n.$ Hence, the dimension of $\Gamma$ is  $[(4K+1)n] \times [(4K+1)n]$.

{Within the framework of quantum unravelling \cite{breuer2002theory}, }
 the density matrix associated with the combined wave functions $\Psi$ in \cref{eq: Psi-V_k} is defined as the entry-wise expectation,
\begin{equation}
    \Gamma_{\alpha,\beta}(t) = \mathbb{E}[\Psi_\alpha(t) \Psi_\beta(t)^*].
\end{equation}
{An application of the It\^o's formula \cite{Pav_book:14} yields} the following close-form quantum master equation,
\begin{equation}\label{eq: qme}
 \partial_t \Gamma = \mathcal{K}(\Gamma):= -i (H \Gamma - \Gamma H^\dagger) + \sum_{k=1}^{K}  V_k \Gamma V_k^\dagger.
\end{equation}
The noise has been averaged out by the expectation. In fact, it has been shown in \cite{kloeden2013numerical} that for any linear SDEs, the first and second moments satisfy close-form equations. {Intuitively, the Lindblad description breaks down when the dynamics of $\rho_S(t)$ exhibits strong memory effect, \ie, it depends on the past history of $\rho_S(t)$, which can not be captured by \cref{eq: lindblad}. On the other hand, the GQME \cref{eq: qme} embodies the history dependence by embedding
the dynamics of $\rho_S(t)$ in a larger system. Conceptually, if one solves other components of $\Gamma(t)$ and substitutes those solutions to the first block, one would get the memory dependence on $\rho_S(t)$.
} 

We can deduce the initial condition of $\Gamma$ from the definitions of the auxiliary wave functions. Since $\zeta_k$ is Gaussian with mean zero and variance 1, we have \[
\mathbb{E}[\braket{\chi_k^\i}{\chi_k^\i}]=0, \;
\mathbb{E}[\braket{\chi_k^\ii}{\chi_k^\ii}]=\rho_s(0), \;  \mathbb{E}[\braket{\chi_k^\iii}{\chi_k^\iii}]=0, \;
\mathbb{E}[\braket{\chi_k^\iiii}{\chi_k^\iiii}]=3\rho_s(0).\]
{All the cross-correlations are zero.}
Therefore, the initial density matrix $\Gamma$ is a block-diagonal matrix,
\begin{equation}\label{eq: gamma(0)}
    \Gamma(0)= \text{diag}\big\{\rho_S(0), 0, \rho_S(0),0, 3\rho_S(0), 0, \rho_S(0), 0, 3\rho_S(0),  \cdots   \big\}.
\end{equation}
We notice that the trace of $\Gamma(0)$ is $4K+1$.

\subsection{Properties of the GQME}
We first provide some basic estimates on the extended density matrix $\Gamma$. We begin by writing the Hamiltonian in \cref{eq; Hnew} as,
\begin{equation}\label{eq: H0H1}
    H= H_0 + \lambda H_1.
\end{equation}
Here we made the observation that $H_0$ is block diagonal. On the other hand, $H_1$ only contains nonzero blocks in the first row and the first column. To conveniently  refer to the block entries of  the density matrix $\Gamma$, we write it in the following block form,
\begin{equation}\label{eq gamma-block}
    \Gamma =\left[
    \begin{array}{cccc}
        \Gamma_{0,0} & \Gamma_{0,1} &\Gamma_{0,2} & \cdots  \\
          \Gamma_{1,0} & \Gamma_{1,1} &\Gamma_{1,2} & \cdots\\
          \vdots & \vdots & \vdots & \ddots \\
           \Gamma_{4K,0} & \Gamma_{4K,1} &\Gamma_{4K,2} & \cdots
    \end{array}\right].
\end{equation}
In particular, $\rho_S$ is embedded into $\Gamma$ as the first block: 
 $\rho_S=\Gamma_{0,0}.$

For such block matrices, we will use the following induced norm,
\begin{equation}\label{eq: inf-norm}
    \norm{ \Gamma }_\infty := \max_{0\le j \le 4K} \sum_{0\le i \le 4K} \norm{\Gamma_{i,j}}.
\end{equation}
Namely, for each entry, we use the spectral norm. But among the blocks, we use the $\infty$-norm.  One can verify that this norm still has the submultiplicative property.  We choose this norm merely because we will estimate the bound of each block, and the formula in \cref{eq: inf-norm} can easily connect such estimates to the bound of the entire matrix. In principle, since matrix norms are continuous with respect to the entries, one can also use other norms among the blocks. 
 
\begin{lemma}
   The exponential operator $U(t):=\exp \big(-it H_0 \big)$ is bounded for all time:
   \[ \| U(t) \| \le 1,  \forall t \in \mathbb{R}.\]
\end{lemma}
Here we used the spectral norm. This can be seen from the fact that $H_0$ is block diagonal,  $H_s$ is Hermitian, and $d_k$ has non-negative imaginary parts.

By separating the $\mathcal{O}(\lambda)$ term in \cref{eq: H0H1}, we can write \cref{eq: qme} in a perturbative form,
\begin{equation}\label{eq: pert}
    \partial_t \Gamma = -i (H_0 \Gamma - \Gamma H_0^\dagger) + \sum_{k=1}^{K}  V_k \Gamma V_k^\dagger - i \lambda (H_1 \Gamma - \Gamma H_1^\dagger).
\end{equation}
In particular, we let $\Gamma^\sz$ be the solution of the ``unperturbed'' equation,
\begin{equation}\label{eq: upert}
    \partial_t \Gamma = -i (H_0 \Gamma - \Gamma H_0^\dagger) + \sum_{k=1}^{K}  V_k \Gamma V_k^\dagger.
\end{equation}
To further simplify notations, let $\CL_0$ be the corresponding operator on the right hand side of \cref{eq: upert}. Then one can express the solution concisely as, 
\begin{equation}\label{eq: CL0}
\Gamma^\sz(t) = \exp (t\CL_0) \Gamma^\sz(0).
\end{equation}

The following Lemma provides a bound for the perturbation term in \cref{eq: pert}.
\begin{lemma}\label{lmm: p-term}
  Let $\Lambda = \max_{1\leq k \leq K} \|T_k\|_1$, and let 
  \begin{equation}\label{eq: Xi}
      \Xi=H_1\Gamma - \Gamma H_1^\dagger,
  \end{equation}
 be the commutator for  any Hermitian matrix $\Gamma$.
  Then the trace of the first diagonal block of $\Xi$ is bounded by,
  \begin{equation}\label{eq: xi00}
      \abs{ \tr \big(\Xi_{0,0} \big) } \leq 4\Lambda  \sum_{k>0}  \| \Gamma_{0,k} \|.
  \end{equation}
For the remaining diagonal blocks, it holds that,
\begin{equation}
    \sum_{k>0}  \abs{ \tr \big(\Xi_{k,k} \big) } \leq 4 \Lambda \sum_{0<j<k\leq 4K}
     \left\| \Gamma_{j,k} \right\|.
\end{equation}
\end{lemma}
These estimates can be obtained by direct calculations. For instance,  the first diagonal 
is given by,
\begin{equation}\label{eq: xi00'}
   \Xi_{0,0} = \sum_{k=1}^K \Big( \Gamma_{0,4k-3} T_k - T_k^\dagger  \Gamma_{4k-3,0} + \Gamma_{0,4k-2} T_k^\dagger - T_k  \Gamma_{4k-2,0} \Big).   
\end{equation}
Thus the bound \cref{eq: xi00} follows from the triangle inequalities, together with the von Neumann's trace inequality.
The important observation is that the trace of the perturbation term in \cref{eq: pert} is only controlled by the norms of the off-diagonal blocks of $\Gamma.$

We now show that the ``unperturbed part'' in \cref{eq: pert}, \ie, the solution of  the GQME in \cref{eq: qme} when $\lambda=0,$ has bounded solutions. The $\lambda>0$ case can then be handled using a perturbation technique.  
\begin{lemma}\label{lmm: upert}
   Assume that the imaginary parts of $d_k$ are non-negative, \ie, $\nu_k \geq  0$ for all $k$. Assume $\lambda=0$. The solution of the GQME in \cref{eq: qme} is denoted by $\Gamma^\sz(t)$ (also see \cref{eq: upert}).   Then the following statements hold.
   \begin{enumerate}
      \item The first diagonal block      is given by, \[ \rho(t)= U_S(t) \rho_S(0)U_S(t)^\dagger, \quad U_S(t):=\exp(-itH_s).\]

       \item If $\Gamma^\sz(0)$ is block diagonal, then $\Gamma^\sz(t)$ remains block diagonal.
       
       \item  If $\Gamma^\sz(0)$ is given by \cref{eq: gamma(0)}, then,  the trace of $\Gamma^\sz(t)$ is 
         \begin{align}
           \label{eq:tr-Gamma}
           \tr\big(\Gamma^\sz(t) \big) = 2K + 1 + 2 \sum_{k} e^{-4\nu_k t}.
         \end{align}
       
         \item  The solution  $\Gamma^\sz(t)$ of \cref{eq: upert} is bounded for general initial conditions. Namely, the exists a constant $c$, independent of $t$, such that, \[ \|\Gamma^\sz(t)\|_\infty \leq c\| \Gamma^\sz(0)\|_\infty, \quad \forall t \in \mathbb{R}_+.\]
   \end{enumerate}
\end{lemma}
We included the proof in \cref{p: upert}.

\medskip

The GQME in \cref{eq: qme} is considered as a route to obtain an approximation to the density matrix $\rho_S(t)$ from \cref{eq: lvn}. Assuming that the representation of the bath correlation function in \cref{eq:correlation} is exact, we can show that the error associated with the approximation of $\rho_S(t)$ by the GQME in \cref{eq: qme} is $\co{\lambda^3}$. 
\begin{theorem}\label{eq: thm-rhos}
Let $\rho_S(t)$ be the density matrix from the full quantum model in \cref{eq: lvn} with bath correlation given by \cref{eq:correlation}. In addition, let $\widehat{\rho}_S(t)$ be the first diagonal block of the density matrix $\Gamma(t)$ from the GQME in \cref{eq: qme}. Then, 
\begin{equation}
     \rho_S(t) = \widehat{\rho}_S(t) + \co{\lambda^3}.
\end{equation}
\end{theorem}
This asymptotic error is consistent with that of the NMSSE in \cref{eq: sse}. 

\medskip 
In the computation, we work with  the GQME in \cref{eq: qme}, and the next theorem provides some bounds on the solution $\Gamma(t)$.
\begin{theorem}\label{eq: thm-gamma}
For any $t>0$, the density matrix $\Gamma(t)$ from the GQME in \cref{eq: qme} is positive semidefinite: $\Gamma(t) \geq 0$, and it has the following properties.
\begin{itemize}
    \item[{\bf (i)}] The norm of the density matrix follows the bound, 
    \begin{equation}
    \|\Gamma(t)\|_\infty \le \|\Gamma(0)\|_\infty \exp (2 \lambda C \|H_1\| t).
\end{equation}
The constant $C$ is the same as that in the Lemma 4. 

\item[{\bf (ii)}] The norms of the  off-diagonal blocks of $\Gamma(t)$ is of order $\lambda$.

\item[{\bf (iii)}] For any initial condition $\Gamma(0),$ not necessarily positive, the trace of $\Gamma(t)$ is bounded as,
\begin{equation}\label{eq: tr-bd}
  \mid    \tr\big(\Gamma(t) \big) \mid 
     \le 3K\mid   \tr\big(\Gamma_{1,1}(0)\big)\mid   
     + \sum_{k>1}\mid   \tr\big(\Gamma_{k,k}(0)\big)\mid  . 
\end{equation}

\end{itemize}
\end{theorem}

\begin{corollary}\label{eq: tr-bound}
Fix $T>0$. The trace of $\Gamma(t)$ from \cref{eq: qme} for $t\in [0,T]$ is bounded as follows,
\begin{equation}
   \tr \left( \Gamma(t) \right) = 2K +1 + 2 \sum_{k=1}^K e^{-4 \nu_k t} +  \mathcal{O} (\lambda^2).
\end{equation}

More importantly, the trace of $\widetilde{\rho}_S(t)$, as the first block diagonal of $\Gamma(t)$, is bounded by,
\begin{equation}
    \tr\left(\widetilde{\rho}_S(t)\right)= \tr\left(\Gamma_{0,0}(t)\right) = 1 + \mathcal{O} (\lambda^3).
\end{equation}

\end{corollary}

\section{Quantum algorithm for simulating generalized quantum master equations}
\label{sec:simulation}

Before presenting the quantum algorithm for this problem, we first prove some useful technical results that will be used in the analysis of our quantum algorithm.

\begin{lemma}[Initial state preparation]
  \label{lemma:state-prep}
  Given a copy of any initial state $\ket{\psi}$ of the system, a normalized version of $\Gamma_0$, as defined in \cref{eq: gamma(0)}, can be prepared using $\CO(K)$ 1- and 2-qubit gates.
\end{lemma}
\begin{proof}
  Let $K' \leq 2K+1$ be the smallest integer larger than $K$ such that $(4K'+1)$ is some power of 2. We initialize a $\log(4K'+1)$-qubit state which is a normalized version of
  \begin{align}
    \ket{0} + \sum_{k=1}^K \Big( \ket{4k-2} + \ket{4k} \Big).
  \end{align}
  This can be implemented using $\CO(K)$ 1- and 2-qubit gates. Append an additional register that is initialized to $\ket{0}$ and use $\mathrm{CNOT}$ gates to ``copy'' the basis states to the second register. We have the normalized version of
  \begin{align}
    \ket{0}\!\ket{0}+ \sum_{k=1}^K \Big( \ket{4k-2}\!\ket{4k-2} + \ket{4k}\!\ket{4k} \Big).
  \end{align}
  Then, simply by appending the original state $\rho_S$ of the system and tracing out the second register, we obtain the normalized version of $\Gamma_0$.
\end{proof}

\noindent{\bf Constructing block-encodings. }
In the next two lemmas, we show how to obtain block-encodings of $H$ and $V_j$  in \cref{eq: qme} for $j \in [K]$ using block-encodings of $H_S$ and $S_\beta$ for $\beta \in [M]$.

\begin{lemma}
  \label{lemma:block-coding-H}
  {
  Let $H$ be defined in \cref{eq:H-block}, and recall $\norm{\theta}_1 = \sum_{k=1}^K\theta_k$. Given the access to an $(\alpha, a, \epsilon)$-block-encoding $U_{H_S}$ of $H_S$, an $(\alpha, a, \epsilon)$-block-encoding $U_{S_\beta}$ for each $\beta \in [M]$, an $(\alpha', a', \epsilon')$-block-encoding
  of $I-i\delta H$, where
  \begin{align}
    \alpha' &\leq 1+\delta\left(\alpha + d_{\max} + 4\lambda\alpha \sqrt{M}\norm{\theta}_1 \right),\\
    a' &= a+\CO(\log(MK)),\\
    \epsilon' &= \alpha'\epsilon.
  \end{align}
  can be implemented with one invocation of $U_{H_S}$ and $2K$ invocations to each of $U_{S_\beta}$ together with $\CO(MK)$ 1- and 2-qubit gates.
}
\end{lemma}
\begin{proof}
  Combining \cref{eq:H-block,eq: T}, we can write $H$ as
  \begin{align}
    \begin{aligned}
      H &= I_A \otimes H_S + H_A \otimes I_S + i\lambda\sum_{k=1}^K\sum_{\beta=1}^M\theta_k\braket{Q_k}{\beta}D_k\otimes S_{\beta} \\
        &\quad\quad + i\lambda\sum_{k=1}^K\sum_{\beta=1}^M\theta_k\braket{\beta}{Q_k}E_k\otimes S_{\beta}^{\dag}.
    \end{aligned}
  \end{align}

  We first show that it is easy to construct a block-encoding of the tensor product of two block-encodings. Let $U_A$ be an $(\alpha_1, a, \epsilon)$-block-encoding of $A$ and $U_B$ be an $(\alpha_2, b, \epsilon)$-encoding of $B$, then it is straightforward to see that the unitary $U_A\otimes U_B$ together with $\CO(a+b)$ swap gates is an $(\alpha_1\alpha_2, a+b, \epsilon)$-block-encoding of $A\otimes B$. As a result, an $(\alpha, a, \epsilon)$-block-encoding of $I_A\otimes H_S$ can be implemented. Since $H_A$ is diagonal, it is easy to implement a $(d_{\max}, \CO(\log K), 0)$-block-encoding of $H_A$. Hence, a $(d_{\max}, \CO(\log K), 0)$-block-encoding of $H_A\otimes I_S$ can be implemented.

  Since $D_k$ is $2$-sparse, we use \cref{lemma:sparse-to-be} to implement an $(2, \CO(\log(K)), 0)$-block-encoding $U_{D_k}$ of $D_k$. Note that this is an exact block-encoding because each nonzero entry of $D_k$ is 1. We also need to implement the sparse-access oracles for $D_k$. Since it only has constant nonzero entries, this can be done by a small classical circuit and turning it into a quantum circuit with $\CO(\log K)$ 1- and 2-qubit gates. Using the observation on tensor products of block-encodings, we can implement a $(2\alpha, a+\CO(\log K), \epsilon)$-block-encoding of $D_k\otimes S_{\beta}$  given an $(\alpha, a, \epsilon)$-block-encoding $U_{S_{\beta}}$ of $S_{\beta}$. $E_k \otimes S_{\beta}^{\dag}$ can be implemented similarly.

  Note that $I-i\delta H$ contains $3+2KM$ terms. They are $I_A\otimes I_S$, $-\delta I_A\otimes H_S$, $-\delta H_A\otimes I_S$, $-i\delta\lambda\theta_k\braket{Q_k}{\beta}D_k\otimes S_{\beta}$, and $-i\delta\lambda\theta_k\braket{Q_k}{\beta}E_k\otimes S_{\beta}^{\dag}$ for each $k,\beta$. The linear combination of normalizing factors is
  \begin{align}
    \label{eq:nom-ub}
    &1+\delta\left(\alpha + d_{\max} + 2\lambda\sum_{k=1}^K\sum_{\beta=1}^M\theta_k \lvert \braket{Q_k}{\beta}\rvert 2\alpha \right) \\
    &\quad\quad \leq 1+\delta\left(\alpha + d_{\max} + 4\lambda\alpha\sqrt{M}\norm{\theta}_1 \right),
  \end{align}
  where we have used the Cauchy–Schwarz inequality. As a result, we can use \cref{lemma:sum-to-be} to construct a block-encoding of $I-i\delta H$ with the desired parameters. This construction uses $2K$ invocations to each of $U_{S_{\beta}}$, $2M$ invocations to the block-encodings of each $D_k$ and $E_k$ (using $\CO(M\log K)$ total 1- and 2-qubit gates), and one invocation to $U_{H_S}$. The additional 1- and 2-qubit gates are used for the state preparation in \cref{lemma:sum-to-be} and this cost is $\CO(MK)$, which dominates the gate cost.

\end{proof}

\begin{lemma}
  \label{lemma:block-coding-V}
  Given $d_k$'s as in \cref{eq:correlation}, a $(d_{\max}, \CO(\log(K)), \epsilon)$-block-encoding for $V_k$ for $k \in [K]$ can be constructed using $\CO(\log K + \polylog(1/\epsilon))$ 1- and 2-qubit gates.
\end{lemma}
\begin{proof}
  Each $V_k$ is 1-sparse, and it is straightforward to implement the sparse-access oracles specified in \cref{eq:sparse-1} and \cref{eq:sparse-2}. Then we use \cref{lemma:sparse-to-be} to implement a $(1, \CO(\log(K)), \epsilon)$-block-encoding for $V_k/d_{\max}$, which implies a $(d_{\max}, \CO(\log(K)), \epsilon)$-block-encoding for $V_k$.
\end{proof}

\noindent{\bf Infinitesimal approximation by completely-positive maps. }
An important step of our quantum algorithm is to use the following superoperator to approximate $e^{\mathcal{K}\delta}$ when $\delta$ is small.
\begin{align}
  \label{eq:m-delta}
  \mathcal{M}_{\delta}(X) = A_0XA_0^{\dag} + \sum_{k=1}^KA_kXA_k^{\dag},
\end{align}
where
\begin{align}
  \label{eq:m-delta-kraus}
    A_0 = I - i\delta H, \text{ and } \quad A_k = \sqrt{\delta}V_k \text{ for all $k \in [K]$}.
\end{align}

We use the following lemma, which is proved in \cref{sec:proof-dist-m-k}, to bound the error of this approximation:
\begin{lemma}
  \label{lemma:dist-m-k}
  Let $\mathcal{M}_{\delta}$ be a superoperator defined in \cref{eq:m-delta}, and let $\mathcal{K}$ be as defined in \cref{eq: qme}. Define $\Lambda:= \norm{H} + \sum_{k=1}^K\norm{V_k}^2$. Then it holds that
  \begin{align}
    \norm{\mathcal{M}_{\delta} - e^{\mathcal{K}\delta}}_{\diamond} \leq 5(\delta \Lambda)^2.
  \end{align}
\end{lemma}

\noindent{\bf Oblivious amplitude amplification for isometries. }
In our algorithm, we need to apply amplitude amplification to boost the success probability. However, in our context, the underlying operator is an isometry instead of a unitary (\ie, part of the input is restricted to be some special state). We use \emph{oblivious amplitude amplification for isometries}, which was first introduced in~\cite{CW17} to achieve this. In \cite{CW17}, only a special case where the initial success probability is exactly $1/4$ was considered. Here, we give a more general version.

\begin{lemma}
  \label{lemma:oaa}
  For any $a, b \in \mathbb{N}_+$, let $\ket*{\widehat{0}} := \ket{0}^{\otimes a}$ and $\ket*{\widehat{\mu}} := \ket{\mu}^{\otimes b}$ for an arbitrary state $\ket{\mu}$. For any $n$-qubit state $\ket{\psi}$, define $\ket*{\widehat{\psi}} := \ket*{\widehat{0}}\ket*{\widehat{\mu}}\ket{\psi}$. Let the target state $\ket*{\widehat{\phi}}$ be defined as
  \begin{align}
    \ket*{\widehat{\phi}} := \ket*{\widehat{0}}\ket{\phi},
  \end{align}
  where $\ket{\phi}$ is a $(b+n)$-qubit state.
  Let $P_0:=\ketbra*{\widehat{0}}{\widehat{0}}\otimes I_{2^b} \otimes I_{2^n}$ and $P_1:=\ketbra*{\widehat{0}}{\widehat{0}}\otimes \ketbra{\widehat{\mu}}{\widehat{\mu}} \otimes I_{2^n}$ be two projectors.
  Suppose there exists an operator $W$ such that
  \begin{align}
    \label{eq:w}
    W\ket*{\widehat{\psi}} = \sin\theta \ket*{\widehat{\phi}} + \cos\theta\ket*{\widehat{\phi}^{\bot}},
  \end{align}
  for some $\theta \in[0, \pi/2]$ with $\ket*{\widehat{\phi}^{\bot}}$ satisfying $P_0\ket*{\widehat{\phi}^{\bot}} = 0$.
 Then it holds that
  \begin{align}
    -\!W\!\left(I\!-\!2P_1\right)W^{\dag}\left(I\!-\!2P_0\right)(\sin\gamma\ket*{\widehat{\phi}}\!+\!\cos\gamma\ket*{\widehat{\phi}^{\bot}}) \!=\! \sin(\gamma\!+\!2\theta)\ket*{\widehat{\phi}}\!+\!\cos(\gamma\!+\!2\theta)\ket*{\widehat{\phi}^{\bot}},
  \end{align}
  for all $\gamma \in [0, \pi/2]$.
\end{lemma}

\begin{proof}
  Let $\ket*{\widehat{\psi}^{\bot}}$ be a state satisfying
  \begin{align}
    \label{eq:w-dag}
    W\ket*{\widehat{\psi}^{\bot}} = \cos\theta\ket*{\widehat{\phi}} - \sin\theta \ket*{\widehat{\phi}^{\bot}}.
  \end{align}
 It is useful to have the following facts
  \begin{align}
    W^{\dag}\ket*{\widehat{\phi}} &= \sin\theta \ket*{\widehat{\psi}} + \cos\theta \ket*{\widehat{\psi}^{\dag}} \quad \text{ and}\\
    W^{\dag}\ket*{\widehat{\phi}^{\bot}} &= \cos\theta \ket*{\widehat{\psi}} - \sin\theta \ket*{\widehat{\psi}^{\dag}},
  \end{align}
  which can be obtained from \cref{eq:w} and \cref{eq:w-dag}.

  We first show that $P_1\ket*{\widehat{\psi}^{\bot}} = 0$. To see this, we define an operator
  \begin{align}
    Q = (\bra*{\widehat{0}}\bra*{\widehat{\mu}}\otimes I) W^{\dag} P_0 W (\ket*{\widehat{0}}\ket*{\widehat{\mu}}\otimes I).
  \end{align}
  For any $\ket{\psi}$, we have
  \begin{equation}
  \begin{aligned}
    \bra{\psi}Q\ket{\psi} &= \norm{P_0W\left(\ket*{\widehat{0}}\ket*{\widehat{\mu}}\ket{\psi}\right)}^2 \\ 
    &= \norm{P_0\left(\sin\theta\ket*{\widehat{\phi}} + \cos\theta\ket*{\widehat{\phi}^{\bot}}\right)}^2 = \norm{\sin\theta \ket*{\widehat{\phi}}}^2 = \sin^2\theta.
  \end{aligned}
    \end{equation}

  Hence, all the eigenvalues of $Q$ are $\sin^2\theta$, so we can write
  \begin{align}
    \label{eq:q}
    Q = \sin^2\theta I.
  \end{align}
  Now, consider any state $\ket{\psi}$:
  \begin{equation}
  \begin{aligned}
    \label{eq:qpsi-1}
    Q\ket{\psi} &= (\bra*{\widehat{0}}\bra*{\widehat{\mu}}\otimes I) W^{\dag} P_0 W (\ket*{\widehat{0}}\ket*{\widehat{\mu}}\otimes \ket{\psi}) = \sin\theta (\bra*{\widehat{0}}\bra*{\widehat{\mu}}\otimes I) W^{\dag}\ket*{\widehat{\phi}} \\
                &= \sin\theta (\bra*{\widehat{0}}\bra*{\widehat{\mu}}\otimes I) \left(\sin\theta\ket*{\widehat{\psi}} + \cos\theta \ket*{\widehat{\psi}^{\bot}}\right)  \\&= \sin^2\theta\ket{\psi} + \sin\theta\cos\theta(\bra*{\widehat{0}}\bra*{\widehat{\mu}}\otimes I)\ket*{\widehat{\psi}^{\bot}},
  \end{aligned}
   \end{equation}
  where the third equality follows from \cref{eq:w} and \cref{eq:w-dag}. On the other hand, by \cref{eq:q}, we have
  \begin{align}
    \label{eq:qpsi-2}
    Q\ket{\psi} = \sin^2\theta\ket{\psi}.
  \end{align}
In light of  \cref{eq:qpsi-1} and \cref{eq:qpsi-2}, we must have, 
  \begin{align}
    \left(\bra*{\widehat{0}}\!\bra*{\widehat{\mu}}\otimes I\right)\ket*{\widehat{\psi}^{\bot}} = 0,
  \end{align}
  which implies that $P_1\ket*{\widehat{\psi}^{\bot}} = 0$.

 We now proceed to analyze the result of applying $W(I-2P_1)W^{\dag}(I-2P_0)$ on $\sin\gamma\ket*{\widehat{\phi}}+\cos\gamma\ket*{\widehat{\phi}^{\bot}}$, 
  \begin{align}
    &W(I-2P_1)W^{\dag}(I-2P_0)\left(\sin\gamma\ket*{\widehat{\phi}}+\cos\gamma\ket*{\widehat{\phi}^{\bot}}\right) \\
    &= \left(I - 2P_0 -2WP_1W^{\dag} + 4WP_1W^{\dag}P_0\right)\left(\sin\gamma\ket*{\widehat{\phi}}+\cos\gamma\ket*{\widehat{\phi}^{\bot}}\right) \\
    &= \sin(\gamma + 2\theta)\ket*{\widehat{\phi}} + \cos(\gamma + 2\theta)\ket*{\widehat{\phi}^{\bot}}.
  \end{align}
\end{proof}

\noindent{\bf Proof of the main theorem. } Now, we have all the tools for proving the quantum algorithm for simulating \cref{prob:simulation}.
\begin{theorem}
  \label{thm:qalg}
  Suppose we are given a block-encoding $U_{H_S}$ of $H_S$, a block-encodings $U_{S_{\alpha}}$ of $S_{\alpha}$ (for $\alpha \in [M]$), $\theta_k$, $\ket{Q_k}$ (all entries), $\lambda$, and $d_k$ for $k \in [K]$ as in \cref{eq:correlation}. There exists a quantum algorithm that solves \cref{prob:simulation}. Let $\tau = t(\alpha(1+\lambda\sqrt{M}\norm{\theta}_1)+Kd_{\max}^2)$, where $\norm{\theta}_1 = \sum_{k=1}^K\theta_k$. This quantum algorithm uses 
  \begin{align}
    \CO\left(\tau\sqrt{K}\frac{\log(\tau/\epsilon)}{\log\log(\tau/\epsilon)}\right),
  \end{align}
  queries to $U_{H_S}$ and $U_{S_a}$, and 
  \begin{align}
    \CO\left(\frac{\tau\sqrt{K}\log^2(\tau/\epsilon)}{\log\log(\tau/\epsilon)} + \tau K^{2.5}\right),
  \end{align}
  additional 1- and 2-qubit gates.
\end{theorem}

  We first use \cref{lemma:state-prep} to prepare a normalized version of $\Gamma_0$ as in \cref{eq: gamma(0)}.
  To simulate the dynamics in \cref{eq: qme}, first note that the solution to \cref{eq: qme} is $e^{\mathcal{K}t}$. We use the superoperator $\mathcal{M}_{\delta}$ defined in \cref{eq:m-delta} and \cref{eq:m-delta-kraus} to approximate $e^{\delta\mathcal{K}}$ for a small step $\delta$. By \cref{lemma:dist-m-k}, we know that the approximation error is at most $5(\delta\Lambda)^2$, where $\Lambda = \norm{H} + \sum_{k=1}^K\norm{V_k}^2$.

  To use \cref{lemma:block-encoding-channel} to implement $\mathcal{M}_{\delta}$, we first need to implement the block-encoding of $A_j$ for $j \in \{0, \ldots, K-1\}$. Using \cref{lemma:block-coding-H}, a $(\alpha', a', \epsilon')$-block-encoding of $A_0$ can be implemented with $\alpha' \leq 1+\delta(\alpha + d_{\max} + 4\lambda\alpha \sqrt{M}\norm{\theta}_1))$, $a' = a+\CO(\log(MK))$, and $\epsilon' = \alpha'\epsilon$. Also, a $(\sqrt{\delta}d_{\max}, \CO(\log(K)), \epsilon)$-block-encoding of $A_j$ can be implemented for each $j \in [K]$ using $\CO(\log K + \polylog(1/\epsilon))$ 1- and 2-qubit gates by \cref{lemma:sparse-to-be}. Now we use \cref{lemma:block-encoding-channel} to obtain the unnormalized state $\sum_{j=1}^m\ket{j}A_j\ket{\psi}$ with the ``success probability'' parameter
  
  \begin{small}
  \begin{align}
    &\geq \frac{1}{\left(1+\delta(\alpha + d_{\max} + 4\lambda\alpha \sqrt{M}\norm{\theta}_1 )\right)^2 + \sum_{j=1}^K\delta d_{\max}^2} \\
    &= \frac{1}{(1+\delta(2\alpha + 2d_{\max} + 8\alpha\lambda\sqrt{M}\norm{\theta}_1+Kd_{\max}^2) + \delta^2(\alpha + d_{\max} + 4\alpha\lambda\sqrt{M}\norm{\theta}_1)^2} \\
    &=1 - \delta(2\alpha + 2d_{\max} + 8\alpha\lambda\sqrt{M}\norm{\theta}_1+Kd_{\max}^2) + \CO(\delta^2(\alpha + d_{\max} + 4\alpha\lambda\sqrt{M}\norm{\theta}_1)^2).
  \end{align}
\end{small}

  Setting the step size, 
  \begin{align}
    \delta = \Theta\left(\frac{1}{r(2\alpha + 2d_{\max} + 8\alpha\lambda\sqrt{M}\norm{\theta}_1+Kd_{\max}^2)}\right),
  \end{align}
  and repeating the above procedure $r$ times, this success probability parameter becomes $\CO(1)$.

  Now for the approximation error, we observe that 
  \begin{equation}
  \begin{aligned}
    &\norm{\mathcal{M}_{\delta} - e^{\delta \mathcal{K}}}_{\diamond} \\
    &\leq 5\delta^2\Lambda^2 
    = \CO\left(\frac{\Lambda^2}{r^2(2\alpha + 2d_{\max} + 8\alpha\lambda\sqrt{M}\norm{\theta}_1+Kd_{\max}^2)^2}\right) \leq \CO\left(\frac{1}{r^2}\right).
  \end{aligned}
    \end{equation}

{By applying the approximation $r$ times,}  we have
  \begin{align}
    \norm{\mathcal{M}_{\delta}^r - e^{t\mathcal{K}}}_{\diamond} \leq \CO\left(\frac{1}{r}\right),
  \end{align}
  for evolution $t = r\delta = \Theta\left(\frac{1}{2\alpha + 2d_{\max} + 8\alpha\lambda\sqrt{M\norm{\theta}_1}+Kd_{\max}^2}\right)$. The parameter $r$ can be chosen larger enough so that this error is at most $\epsilon$.

  Further, conditioned on that we have obtained an approximation $\widetilde{\Gamma}_t$ of $\Gamma_t$, we can extract an approximation of $\ketbra{\psi}{\psi}$ by measuring the first $\log(4K+1)$ qubits of $\widetilde{\Gamma}_t$ and post-select the outcome being 0. According to \cref{eq: tr-bound}, this probability of the outcome being 0 is $\Omega(1/(2K+1))$. Now, the success probability is $\Omega(1/K)$. It follows from \cref{lemma:oaa} that using $\CO(\sqrt{K})$ iterations of oblivious amplitude amplification for isometries, we obtain the desired state.

  Now, the total evolution time we have simulated so far is $r\delta = \Theta\left(\frac{1}{2\alpha + 2d_{\max} + 8\alpha\lambda\sqrt{M}\norm{\theta}_1+Kd_{\max}^2}\right)$, and the total cost in terms of queries to $U_{H_S}$ and $U_{S_\beta}$ is $\CO(r)$. In the following, we show how to achieve poly-logarithmic cost.
\smallskip 
Note that when applying \cref{lemma:sum-to-be} to construct a block-encoding of $A_0 = I - i\delta H$ and applying \cref{lemma:block-encoding-channel} to implement the superoperator specified by Kraus operators $A_0, A_1, \ldots, A_K$, the first register (containing $\CO(\log K)$ qubits) is used for the $\ket{\mu}$ state in \cref{lemma:block-encoding-channel}, and the second register (containing $\CO(\log K)$ qubits) is used for state $B\ket{0}$ in \cref{lemma:sum-to-be}. The coefficients of the state in the two registers are concentrated to $\ket{0}\!\ket{0}$, which corresponds to $I$ (nothing to implement). More specifically, recall that the parameter $s_0$ in \cref{lemma:block-encoding-channel} is the block-encoding normalization factor for $I-i\delta H$, which can be written as  $s_0 = \sum_{j=1}^{2KM+3}y_i\alpha_j$, where $y_j$'s and $\alpha_j$'s are the parameters as used in \cref{lemma:sum-to-be}. 
For convenience, let $\widehat{\alpha}$ be upper bounded by $\widehat{\alpha} \leq \alpha + d_{\max} + 4\alpha\lambda\sqrt{M}\norm{\theta}_1$, so we can write $\sum_{j=1}^{2KM+3}y_i\alpha_j = 1+\delta\widehat{\alpha}$ according to \cref{lemma:block-coding-H}.
  As a result, the amplitude for $\ket{0}\ket{0}$ is
  \begin{align}
    &\frac{s_0}{\sqrt{\sum_{j=0}^Ks_j^2}}\cdot\frac{1}{\sqrt{\sum_{j=1}^{2KM+3}y_j\alpha_j}} = \sqrt{\frac{\sum_{j=1}^{2KM+3}y_j\alpha_j}{\sum_{j=0}^Ks_j^2}} \\
    &= \sqrt{\frac{1+\delta\widehat{\alpha}}{(1+\delta\widehat{\alpha})^2 + Kd_{\max}^2\delta}} \\
    &= \sqrt{\frac{1+\delta(\widehat{\alpha})}{1+2\delta\widehat{\alpha} +\delta Kd_{\max}^2 \!+\! \delta^2\widehat{\alpha}^2}} \\
    &=\sqrt{1-\delta(\widehat{\alpha} + Kd_{\max}^2 ) + \Theta(\delta^2(2\widehat{\alpha}+Kd_{\max}^2)^2)}.
  \end{align}

  Therefore, the probability that the first two registers are \emph{not} measured $0,0$ is proportional to\footnote{Note that we use the term ``proportional to'' because the actual probability is normalized according to the trace ratio of the first block and the whole matrix of $\Gamma$, which incurs a factor $O(1/K)$.}

  \begin{align}
    \delta\!\left(\widehat{\alpha} + Kd_{\max}^2 \right) + \Theta\left(\delta^2(2\widehat{\alpha}+Kd_{\max}^2)^2\right) = \CO\left(\frac{1}{r}\right).
  \end{align}
  We apply the techniques used in~\cite{CW17} (first introduced in~\cite{BCKS14}) to bypass the state preparation for $\ket{\mu}$ in \cref{lemma:sum-to-be} and for $B\ket{0}$ in \cref{lemma:sum-to-be}. Instead, we use a state with Hamming weight at most
  \begin{align}
    h = \CO\left(\frac{\log(1/\epsilon)}{\log\log(1/\epsilon)}\right),
  \end{align}
  to approximate the state after the state preparation procedure while causing error at most $\epsilon$. Following similar analysis as in~\cite{CW17}, the number of 1- and 2-qubit gates for implementing this compressed encoding procedure is $\CO(\log(1/\epsilon)h + K^2)$.
  Now, the number of queries to $U_{H_S}$ and $U_{S_\beta}$ becomes
  \begin{align}
    \CO\left(\frac{\sqrt{K}\log(1/\epsilon)}{\log\log(1/\epsilon)}\right).
  \end{align}
  For arbitrary simulation time $t$, we repeat this   {$\CO(t(2\alpha + 2d_{\max} + 8\alpha\lambda\sqrt{M}\norm{\theta}_1+Kd_{\max}^2)) = \CO(t(\alpha(1+\lambda\sqrt{M}\norm{\theta}_1)+Kd_{\max}^2))$}
  times where each segment has the error parameter $\epsilon/(t(\alpha(1+\lambda\sqrt{M}\norm{\theta}_1)+Kd_{\max}^2))$, which has total cost
  \begin{align}
    \CO\left(\tau\sqrt{K}\frac{\log(\tau/\epsilon)}{\log\log(\tau/\epsilon)}\right).
  \end{align}
  Here we have   {$\tau = t(\alpha(1+\lambda\sqrt{M}\norm{\theta}_1)+Kd_{\max}^2)$.}

  The additional 1- and 2-qubit gates are used in the compressed encoding procedure and the reflections in the oblivious amplitude amplification for isometries. This is bounded by
  \begin{align}
    \CO\left(\tau\sqrt{K}\left(\log(\tau/\epsilon)h + K^2\right)\right) = \CO\left(\frac{\tau\sqrt{K}\log^2(\tau/\epsilon)}{\log\log(\tau/\epsilon)} + \tau K^{2.5}\right).
  \end{align}

\section{Acknowledgments}

 CW thanks Yudong Cao and Peter D.~Johnson for helpful discussions on the HEOM approach for modeling non-Markovian open quantum systems. XL's research is supported by the National Science Foundation Grants DMS-2111221.

\bibliographystyle{alpha}
\bibliography{ref,gqme}

\appendix

\section{The derivation of the extended stochastic dynamics \cref{eq: ext-iv}}\label{a: deriv}

 As a building block to approximate  stationary Gaussian processes  { in the NMSSE (\cref{eq: sse}),}  
 we consider the Ornstein–Uhlenbeck (OU) type of processes \cite{risken1984fokker}, expressed as the solution of the following linear SDEs, 
 \begin{equation}\label{eq: cou}
    i \dot{\zeta}_k = - d_k \zeta_k + 
    \gamma_k \dot{w}_j(t), 
\end{equation}
for $k=1,2,\ldots, K$, with $K \geq M$. Here $\gamma_k\ge 0$ and $d_k$ is a complex number with positive imaginary part.  In addition, each of these independent OU processes has an initial variance $\theta_k^2$, 
 $$\mathbb{E}[\zeta_k(0)^\dagger \zeta_k(0)] = \theta_k^2.$$
If we pick $\gamma_k$,  such that $\gamma_k^2= 2\text{Im}(d_k),$
then $\zeta_k(t) $ is a stationary Gaussian process with correlation,
\begin{equation}\label{eq: corr-zeta}
    \mathbb{E}[\zeta_j^\dagger(t) \zeta_k(t')]= \theta_j^2 \delta_{j,k} \exp \big( -id_j^* (t- t') \big).  
\end{equation}

 We now construct an ansatz for approximating the bath correlation function. The case when $M=1$ has been thoroughly investigated in \cite{ritschel2014analytic}. As a time correlation function, $C(t)$ can be expressed in terms of its power spectrum, denoted here by $\widehat{C}(\omega),$ as a Fourier integral,
\begin{equation}\label{eq: Ct-Fourier}
    C(t) = \int_{-\infty}^{\infty} \widehat{C}(\omega) e^{-i\omega t} d\omega.
\end{equation}
 { Known as the power spectrum, $\widehat{C}(\omega)$ is Hermitian and positive semidefinite.} Hence it can be diagonalized using the transformation:  {
\[ \widehat{C}(\omega) = \sum_{j} \lambda_j(\omega)^2 \ket{Q_j(\omega)}  \bra{Q_j(\omega)}.\] 
Therefore, a direct numerical approximation of \cref{eq: Ct-Fourier} will certainly lead to an ansatz like \cref{eq:correlation}.

Another alternative is to use a contour integral in the upper half plane, and the Cauchy residue theorem, reducing  the Fourier integral in \cref{eq: Ct-Fourier} to a summation over poles \cite{ritschel2014analytic}. This approach will also lead to the ansatz like \cref{eq:correlation}.}

To ensure the approximation accuracy, in practice, the ansatz in \cref{eq:correlation} is often obtained by a least-squares approach. In addition, we will consider the poles $d_k$ within a cut-off frequency $d_\text{max}$, \ie,  {$\abs{d_k} \leq d_\text{max}$.}

 {
Next we show that we can find an approximation of  $\eta(t)$ that exactly satisfies the relation in \cref{eq: fdt}, where $C(t)$ is represented by \cref{eq:correlation}. }
Specifically, we approximate the noise $\eta(t)$ using the OU processes  $\zeta_k(t)$ from \cref{eq: cou},   
\begin{equation} \label{eq: eta-color}
    \eta_\beta(t)= \sum_{k=1}^K \braket{Q_k}{\beta} \theta_k \zeta_k(t).
\end{equation}
In light of \cref{eq: corr-zeta}, the approximation in \cref{eq: eta-color}  corresponds to an approximation of the time correlation function in \cref{eq:correlation}.

For the case of a  single  interaction term, where $M=1$ and the matrix $C(t)$ is corresponding to a scalar function, 
this reduces to the standard approach using a sum of exponentials \cite{ritschel2014analytic}. But \cref{eq:correlation} provides a more general scheme to handle multiple interaction terms.

Next we show how this can simplify the NMSSE in \cref{eq: sse} when the bath correlation functions are expressed as \cref{eq:correlation}. We first define 
the operators $T_j$ according to \cref{eq: T}. 
This simplifies  the NMSSE in \cref{eq: sse} to,
\begin{equation}\label{eq: sse2}
    i \partial_t \psi = \hat H_S \psi - i \lambda^2\sum_{k=1}^K  \int_0^t  T^\dagger_k  
    e^{-i (\hat H_S + d_k^* )\tau } T_k  \psi(t-\tau) d\tau +\lambda \sum_{j=1}^K   \zeta_j(t)  T_j  \psi(t).
\end{equation}
Here the multiplication by $ d_k^*$ represents the operator $d_k^*I_S$,
where $I_S$ is the identity matrix.

\cref{eq: sse2} still contains memory. But compared to the original NMSSE in \cref{eq: sse}, the correlation $C(t)$ has been broken down to exponential functions, which can be combined with the unitary operator $e^{-i t \hat H_S}.$ In addition, the noise is expressed as the OU process $\zeta_j$'s, which can be treated using It\^o calculus. 

Next, we demonstrate how to embed the dynamics in \cref{eq: sse2} into an extended, but Markovian, dynamics.  
The simple observation that motivated the Markovian embedding is that a convolution integral in time can be represented as the solution of a differential equation:
\begin{equation}
 f(t)= \int_0^t \exp(-a(t-\tau)) g(\tau) d\tau \quad \Longrightarrow   f' = -a f + g, \; f(0)=0.
\end{equation}
Here $f$ will be regarded as an auxiliary variable, introduced to reduce the memory integral. Compared to computing the integral at every time step using direct quadrature formulas, it is much more efficient to solve the differential equation. 

We now show that we can use the same idea to define auxiliary orbitals. Specifically, by inserting \cref{eq:correlation}  in the SSE in \cref{eq: sse}, and by letting,  for $k=1,2,\ldots, K,$
\[ \chi_k^\i(t) =  \int_0^t 
    e^{-i (\hat H_S + d_k^*) \tau }   T_k \psi(t-\tau) d\tau,\]
we arrive at the equation,
\begin{equation}\label{eq: ch1}
    i\partial_t \chi_k^\i = (H_S + d_k^* ) \chi_k^\i + i  \lambda T_k \psi(t).  
\end{equation}
This simplifies \cref{eq: sse2} 
to,
\begin{equation}\label{eq: psi-1}
    i \partial_t \psi = \hat H_S \psi - i \lambda^2 \sum_{j=1}^K T_j^\dagger \chi^\i_j + \lambda \sum_{j=1}^K T_j \psi(t) \zeta_j(t).
 \end{equation}
 
To incorporate the noise term, we define,
 \begin{equation}\label{eq: chi2}
     \chi^\ii_k(t)= i\psi(t) \zeta_k(t).
 \end{equation}
This reduces \cref{eq: psi-1} to,
\begin{equation}\label{eq: psi-2}
    i \partial_t \psi = \hat H_S \psi - i \lambda^2 \sum_{j=1}^K T_j^\dagger \chi^\i_j -i  \lambda \sum_{j=1}^K T_j \chi_j^\ii(t).
 \end{equation}

It remains to derive a closed-from equation for \cref{eq: chi2}. Using the It\^o's formula, we obtain, 
 \begin{equation}\label{eq: ito}
      i\partial_t \chi^\ii_k =  (\hat H_S - d_k) \chi^\ii_k   +i  \gamma_k \psi(t) \dot{w}_k + \lambda \zeta_k(t) T_k^\dagger \chi^\i_k + \lambda \zeta_k(t) T_k \chi^\ii_k.
 \end{equation}
 When the coupling parameter $\lambda $ is sufficiently small, one can add or drop $\mathcal{O}(\lambda)$ terms, which will contribute to an  $\mathcal{O}(\lambda^2)$ error when substituted into \cref{eq: psi-2}. If such an error is acceptable, by
 collecting equations, we have extended Schr\"odinger equations,
\begin{equation}\label{eq: ext-ii}
\left\{
    \begin{array}{l}
        i \partial_t \psi =\dsp  \hat H_S \psi - i \lambda    \sum_{k=1}^K T_k^\dagger  \chi^\i_k - i \lambda   \sum_k T_k \chi^\ii_k, \\
    \begin{array}{ll}
            i\partial_t \chi^\i_k =& (\hat H_S + d^*_k ) \chi^\i_k + i \lambda T_k \psi(t),\\
       i\partial_t \chi^\ii_k = & (\hat H_S - d_k) \chi^\ii_k  + i \lambda T_k^\dagger  \psi(t)  +  i  \gamma_k \psi(t) \dot{w}_k.
    \end{array}
    \qquad       k=1,2,\ldots, K.\end{array}
     \right.
\end{equation}
 
Rather than dropping $\mathcal{O}(\lambda)$ terms in \cref{eq: ito}, one can continue such a procedure and incorporate the high-order terms. Specifically, noticing the similarity between the $\mathcal{O}(\lambda)$ terms in \cref{eq: ito} with \cref{eq: chi2}, we define,
\begin{equation}\label{eq: chi34}
        \chi_k^\iii= i \zeta_k(t) \chi_k^\i, \quad  \chi_k^\iiii= i \zeta_k(t) \chi_k^\ii.
\end{equation}
By repeating the above procedure, one can derive similar equations for these auxiliary wave functions. These embedding steps yield the  extended Schr\"odinger equations (ESE) \cref{eq: ext-iv}.

\section{The proof of \cref{lmm: upert}    }\label{p: upert}
\begin{proof}
Our proof will mainly target statement (4). The rest of the lemma will become self-evident throughout the proof. 
In light of the structure of the GQME in \cref{eq: qme}, it is enough to consider the case $K=1.$ In this case, $\Gamma$ can be viewed as a $5\times 5$ block matrix. We first write $V_1$ in a block matrix form,
\[V_1= \left(\begin{array}{cc}
   0  & 0  \\
  R_1 & 0  
\end{array} \right), \] 
where $R_1$ is a $3\times 3$ block matrix with diagonals $\sqrt{2\nu_1}I_S, \sqrt{2\nu_1}I_S$ and $2\sqrt{2\nu_1}I_S$. With direct calculations, we can show that the last term in \cref{eq: qme} can be written as,
\[V_1\Gamma V_1^\dagger  = \left(\begin{array}{cc}
   0  & 0  \\
  0 & R_1 \Gamma_{0:2,0:2} R_1^\dagger   
\end{array} \right). \] 
Here $\Gamma_{0:2,0:2}$ refers to the first $3\times3$ sub-matrix of $\Gamma$. 

We first look at the scenario when $\Gamma(0)$ is block diagonal.  
The zero blocks in $V_1\Gamma V_1^\dagger$, along with the observation that $H_0$ is block diagonal, imply that the off-diagonals do not change. For the diagonal blocks of $\Gamma(t)$, we first have,
\begin{equation}\label{eq: g00-11}
\begin{aligned}
     \Gamma_{0,0}(t) = & U_S(t) \rho_S(0) U_S(t)^\dagger,\\ 
\Gamma_{1,1}(t)= &\exp(-it(-d_1+d_1^*) )  U_S(t) \Gamma_{1,1}(0) U_S(t)^\dagger.
\end{aligned}
\end{equation}
As a result, these two blocks have norms given respectively by,
\[ \| \Gamma_{0,0}(t) \| = \| \Gamma_{0,0}(0) \|, \quad   \| \Gamma_{1,1}(t) \| = \| \Gamma_{1,1}(0) \| e^{-2\nu_1 t}. \] 
The next three block diagonals will pick up non-homogeneous terms. For instance, we have,
\[ \partial_t \Gamma_{2,2}= -i[H_s - d_1, \Gamma_{2,2}] + 2 \nu_1 \Gamma_{0,0}. \]
Using the variation-of-constant formula, we have,
\begin{equation}\label{eq: g22}
    \Gamma_{2,2}(t) = e^{-2\nu_1 t} U_S(t) \Gamma_{2,2}(0) U_S(t)^\dagger 
+2 \nu_1 \int_0^t e^{-2\nu_1 
\tau}  U_S(\tau ) \Gamma_{0,0}(t-\tau) U_S(\tau )^\dagger d\tau.
\end{equation}

Also by noticing that $\|\Gamma_{1,1}(t)\|$ is constant in time,  the diagonal block $\Gamma_{3,3}$ can be bounded directly as,
\begin{equation}
   \| \Gamma_{2,2}(t)\| \leq \| \Gamma_{2,2}(0)\| \exp(-2\nu_1 t) + \|\Gamma_{0,0}(0)\| (1 - e^{-2\nu_1 t}).
\end{equation}
The right hand side remains bounded for all time. 
Similarly, the next diagonal block can be expressed as, 
\begin{equation}\label{eq: g33}
\begin{aligned}
    \Gamma_{3,3}(t) &=\exp(-4\nu_1 t) U_S(t) \Gamma_{3,3}(0) U_S(t)^\dagger 
\\
&+2 \nu_1 \int_0^t \exp(-4\nu_1 t) U_S(\tau ) \Gamma_{1,1}(t-\tau) U_S(\tau )^\dagger d\tau.
\end{aligned}
\end{equation}
Notice that $\| \Gamma_{1,1}(t)\|$ is proportional to $\exp(-2\nu_1 t)$. Essentially, what leads to the boundedness of the solution is the fact that there is no secular term, implying that $\| \Gamma_{3,3}(t)\|$ follows a similar bound as $\| \Gamma_{2,2}(t)\|_2$. The estimate of $\| \Gamma_{4,4}(t)\|$ follows the same steps. 

\medskip

We now turn to the off-diagonal blocks. By direct calculations, we have, for $j=0,1$,  $k=1, 2, 3, 4$,  and $k>j$
\[ \partial_t \Gamma_{j,k} = -i \big(H_S \Gamma_{j,k}-\Gamma_{j,k}(H_s +d_1) \big),\]
which yields,
\begin{equation}
    \Gamma_{j,k}(t) = \exp(-i d_1^* t) U_S(t)\Gamma_{j,k}(0) U_S^{\dagger}(t)
\Longrightarrow \|\Gamma_{j,k}(t)\| = \|\Gamma_{j,k}(0) \| \exp(-\nu_1 t).
\end{equation}
For the remaining off-diagonal entries, we will check $\Gamma_{2,3}$ as an example. It follows the equation,
\[ \partial_t \Gamma_{2,3}
= -i \big((H_s-d_1)\Gamma_{2,3} - \Gamma_{2,3} (H_s + d_1 + d_1^*)\big) + 2\nu_1 \Gamma_{0,1}. \]
This implies that,
\[ \|\Gamma_{3,4}(t) \| \leq \|\Gamma_{3,4}(0) \| e^{-3\nu_1 t} 
+ \|\Gamma_{1,2}(0) \| 2 e^{-2\nu_1 t}(1 - e^{-\nu_1 t}). \]
We also see from these calculations that these off-diagonal blocks will become zero if the initial matrix $\Gamma(0)$ is block diagonal.

By examining the block entries of $\Gamma(t)$, we have shown the boundedness of the solution stated in \cref{lmm: upert} for all time. 
\end{proof}

\section{The Proof of \cref{eq: thm-rhos}}
\begin{proof}
We will prove the asymptotic bound using an expansion of the \cref{eq: lvn}. More specifically, we write the total density matrix in terms of powers of $\lambda,$
\begin{equation}\label{eq: exp-rho}
    \rho(t) = \rho^{(0)}(t) + \lambda  \rho^{(1)}(t) + \lambda^2 \rho^{(2)}(t) + \mathcal{O}(\lambda^3).
\end{equation}
By taking a partial trace over the bath space, we obtain a similar expansion for 
$\rho_S:$
\begin{equation}\label{eq: exp-rhoS}
    \rho_S(t) = \rho_S^{(0)}(t) + \lambda  \rho_S^{(1)}(t) + \lambda^2 \rho_S^{(2)}(t) + \mathcal{O}(\lambda^3).
\end{equation}

By inserting \cref{eq: exp-rho} into \cref{eq: lvn} and separate terms of different order, we arrive at
\begin{equation}\label{eq: rho123}
    \begin{aligned}
         i\partial_t \rho^{(0)} &= [H_S \otimes I_B + I_S \otimes H_B, \rho^{(0)} ], \;\;  \rho^{(0)}(0)=\rho(0), \\ 
          i\partial_t \rho^{(1)} &= [H_S \otimes I_B + I_S \otimes H_B, \rho^{(1)} ] + \sum_{\alpha=1}^M [ S_\alpha \otimes B_\alpha, \rho^{(0)}], \;\; \rho^{(1)}(0) =0, \\ 
           i\partial_t \rho^{(2)} &= [H_S \otimes I_B + I_S \otimes H_B, \rho^{(2)} ]+ \sum_{\alpha=1}^M [ S_\alpha \otimes B_\alpha, \rho^{(1)}], \;\; \rho^{(2)}(0)=0.
    \end{aligned}
\end{equation}
Within this expansion, the dynamics of $\rho^{(0)}$ contains no coupling. Let 
\[
U(t)= U_S(t) \otimes U_B(t), \quad  U_S= \exp \left( -i t H_S\right), \; U_B(t)= \exp \left( -i t H_B\right),\]
be the unitary operators. Then we have,
\begin{equation}
    \rho^{(0)}(t) = U(t) \rho^{(0)}(0) U(t)^\dagger.  
\end{equation}
Since all the operators on the right hand side are in tensor product forms, $\rho^{(0)}(t)$ remains a tensor product:
\begin{equation}\label{eq: rho0t}
     \rho^{(0)}(t) =  \rho_S^{(0)}(t) \otimes \rho_B.
\end{equation}
Here $\rho_S^{(0)}(t) = U_S(t) \rho_S(0) U_S(t)^\dagger.$ Meanwhile, the matrix $\rho_B$ stays because it commutes with $H_B.$ 

The term $\rho^{(1)}$ can be expressed using the variation-of-constant formula,
\[
\begin{aligned}
\rho^{(1)}(t) =& -i \sum_\alpha \int_0^t U(t-t') S_\alpha \otimes B_\alpha  \rho^{(0)}(t')    U(t-t')^\dagger dt' \\ &+ i  \sum_\alpha \int_0^t U(t-t')  \rho^{(0)}(t')  S_\alpha \otimes B_\alpha    U(t-t')^\dagger dt'.      
\end{aligned}
\]
In light of \cref{eq: rho0t}, we can make the same observation that $\rho^{(1)}(t)$ consists of terms that are tensor products. By following standard notations \cite{breuer2002theory}, \ie,
\begin{equation}
    S_\alpha(t)= U_S(t)^\dagger S_\alpha  U_S(t), \quad   B_\alpha(t)= U_B(t)^\dagger B_\alpha  U_B(t),
\end{equation}
we can simplify $\rho^{(1)}(t)$ as follows,
\[
\begin{aligned}
\rho^{(1)}(t) =& -i \sum_\alpha \int_0^t  S_\alpha(t'-t)\rho_S^{(0)}(t)   \otimes B_\alpha(t'-t) \rho_B dt' \\ &+ i  \sum_\alpha \int_0^t  \rho_S^{(0)}(t)   S_\alpha(t'-t)  \otimes B_\alpha(t'-t)  \rho_B dt'.      
\end{aligned}
\]
Since $\rho_B$ commutes with $H_B$, it commutes with $U_B.$ Therefore,
\[ \tr(B_\alpha(t'-t) \rho_B)= \tr(B_\alpha \rho_B) =0,\]
which shows that 
\[  \tr_B \left( \rho^{(1)}(t) \right) =0.\]
Therefore, $\rho_S^{(1)}(t)=0$. The correction to $\rho_S$ comes from $\rho_S^{(2)}(t),$ which is similarly expressed as,
\[
\begin{aligned}
\rho^{(2)}(t) \!=& -i \sum_\alpha \int_0^t U(t-t') S_\alpha \otimes B_\alpha  \rho^{(1)}(t')    U(t-t')^\dagger dt' \\ &+ i  \sum_\alpha \int_0^t U(t-t')  \rho^{(1)}(t')  S_\alpha \otimes B_\alpha    U(t-t')^\dagger dt',\\
=&\! -\! \sum_\alpha \!\sum_\beta \!
\int_0^t \int_0^{t'}  S_\alpha(t'-t) S_\beta(\tau-t) \rho_S^{(0)}(t) \!\otimes \!
 B_\alpha(t'-t) B_\beta(\tau-t) \rho_B d\tau dt',\\
 &\!+\! \sum_\alpha \!\sum_\beta \!
\int_0^t \int_0^{t'} S_\beta(\tau-t) \rho_S^{(0)}(t)  S_\alpha(t'-t)   \!\otimes \!
 B_\beta(\tau-t) \rho_B  B_\alpha(t'-t)  d\tau dt',\\
 &\!+\! \sum_\alpha \!\sum_\beta \!
\int_0^t \int_0^{t'}  S_\alpha(t'-t)  \rho_S^{(0)}(t)  S_\beta(\tau-t) \!\otimes \!
 B_\alpha(t'-t) \rho_B  B_\beta(\tau-t) d\tau dt',\\
 &\!-\! \sum_\alpha \!\sum_\beta \!
\int_0^t \int_0^{t'} \rho_S^{(0)}(t)  S_\beta(\tau-t)  S_\alpha(t'-t)   \!\otimes \!
 \rho_B  B_\beta(\tau-t)  B_\alpha(t'-t)  d\tau dt'.\\
\end{aligned}
\]
Invoking the bath correlation function,
\begin{equation}
    C_{\alpha,\beta}(t)= \tr(B_\alpha(t) B_\beta \rho_B),
\end{equation}
we arrive at an expansion of $\rho_S(t)$ up to $\co{\lambda^2}$,
\begin{equation}\label{eq: rhoS-exp}
\begin{aligned}
   \rho_S(t) = \rho_S^{(0)}(t) & - \lambda^2  \sum_\alpha \sum_\beta 
\int_0^t \int_0^{t'}  S_\alpha(t'-t) S_\beta(\tau-t) \rho_S^{(0)}(t) C_{\alpha,\beta}(t'-\tau) d\tau dt'\\
&+  \lambda^2  \sum_\alpha \sum_\beta  \int_0^t \int_0^{t'} S_\beta(\tau-t) \rho_S^{(0)}(t)  S_\alpha(t'-t) C_{\alpha,\beta}(t'-\tau) d\tau dt' \\
&+ \lambda^2 \sum_\alpha \sum_\beta 
\int_0^t \int_0^{t'}  S_\alpha(t'-t)  \rho_S^{(0)}(t)  S_\beta(\tau-t)   C_{\beta,\alpha}(t'-\tau)^* d\tau dt' \\ 
&- \lambda^2  \sum_\alpha \sum_\beta \int_0^t \int_0^{t'} \rho_S^{(0)}(t)  S_\beta(\tau-t)  S_\alpha(t'-t)  C_{\alpha,\beta}(t'-\tau)^* d\tau dt' \\
& + \mathcal{O}(\lambda^3).
\end{aligned}
\end{equation}
Here we have used the property of the bath correlation function: $C_{\alpha,\beta}(t) = C_{\beta,\alpha}(-t)^*.$
Now we incorporate the function form of the bath correlation function in \cref{eq:correlation}. We find that,
\begin{equation}\label{eq: rhoS-exp'}
\begin{aligned}
   \rho_S(t) = \rho_S^{(0)}(t) & - \lambda^2  \sum_k 
\int_0^t \int_0^{t'}  T_k(t'-t) T_k(\tau-t) \rho_S^{(0)}(t) e^{-i(t'-\tau)d_k t } d\tau dt'\\
&+  \lambda^2  \sum_k  \int_0^t \int_0^{t'} T_k(\tau-t) \rho_S^{(0)}(t)  T_k(t'-t) e^{-i(t'-\tau)d_k t } d\tau dt' \\
&+ \lambda^2  \sum_k 
\int_0^t \int_0^{t'}  T_k(t'-t)  \rho_S^{(0)}(t)  T_k(\tau-t)   e^{i(t'-\tau)d_k t }d\tau dt' \\ 
&- \lambda^2   \sum_k \int_0^t \int_0^{t'} \rho_S^{(0)}(t)  T_k(\tau-t)  T_k(t'-t)  e^{i(t'-\tau)d_k t } d\tau dt' \\
& + \mathcal{O}(\lambda^3).
\end{aligned}
\end{equation}

Now we show that the GQME in \cref{eq: qme} has an asymptotic expansion that is consistent with \cref{eq: rhoS-exp}. Expanding $\Gamma$ as,
\[\Gamma(t)=  \Gamma^{(0)}(t) + \lambda \Gamma^{(1)}(t) + \lambda^2 \Gamma^{(2)}(t) + \mathcal{O}(\lambda^3),  \] 
and substituting it into \cref{eq: qme}, one gets, 
\begin{equation}\label{eq: G-exp}
\begin{aligned}
\Gamma^{(0)}(t) &= \exp\left( t\mathcal{L}_0 \right)  \Gamma(0), \\
\Gamma^{(1)}(t) &= -i \int_0^t \exp\left( (t-t')\mathcal{L}_0 \right)
[H_1, \Gamma^{(0)}(t')] dt',\\
\Gamma^{(2)}(t) &= -i \int_0^t \exp\left( (t-t')\mathcal{L}_0 \right)
[H_1, \Gamma^{(1)}(t')] dt'.
\end{aligned}
\end{equation}
The leading term $\Gamma^{(0)}(t)$ has been shown to be a block diagonal matrix in the previous section. The first diagonal block is precisely $ \rho_S^{(0)}(t)$, which is consistent with the $\co{1}$ term in \cref{eq: rhoS-exp}. To examine $\Gamma^{(1)}(t)$, we first notice that the commutator in the integral has the following structure,
\begin{equation}\label{eq Xi0-block}
  \Xi^{(0)}(t') = [H_1, \Gamma^{(0)}(t')]= i
    \left[
    \begin{array}{cccccc}
         0  & \Xi_{0,1}^{(0)}(t') &\Xi_{0,2}^{(0)}(t') & 0& 0 & \cdots  \\
          \Xi_{1,0}^{(0)}(t') & 0 &0 &  0& 0 & \cdots\\
          \Xi_{2,1}^{(0)}(t') & 0 &0 & \Xi_{3,4}^{(0)}(t') & \Xi_{3,5}^{(0)}(t') & \cdots\\
           0 & 0 & \Xi_{4,3}^{(0)}(t') & 0 &0  & \cdots \\ 
            0 & 0 & \Xi_{5,3}^{(0)}(t') & 0 &0  & \cdots \\ 
          \vdots & \vdots & \vdots &\vdots & \vdots & \ddots \\
    \end{array}\right].
\end{equation}
Here we highlighted the leading $5\times 5$ submatrix and the zero blocks within it. This is enough for the purpose of the proof. 

By inspecting the solutions that correspond to $\exp\left( t\mathcal{L}_0 \right)$, we find that the first diagonal block of $\int_0^t \exp\left( (t-t')\mathcal{L}_0 \right) \Xi^{(0)}(t')$  is zero. Therefore,  $\Gamma^{(1)}$ has no contribution to the density matrix $\rho_S$, \ie, there is no $\co{\lambda}$ term. This is consistent with \cref{eq: rhoS-exp}.

To proceed further, we have to identify the nonzero blocks in \cref{eq Xi0-block}. With direct calculations, we have, 
\[
\begin{aligned}
\Xi_{0,4k-3}^{(0)}(t) =& - \Gamma_{1,1}^{(0)}(t) T_j = -  \rho_{S}^{(0)}(t) T_j, \\  
\Xi_{0,4k-2}^{(0)}(t) =&  T_k \Gamma_{4k-2,4k-2}^{(0)}(t).
\end{aligned}
\] 
From \cref{eq: G-exp}, we can extract the equation,
\[
i \partial_t \Gamma_{0,4k-3}^{(1)} = H_S\Gamma_{0,4k-3}^{(1)} - \Gamma_{0,4k-3}^{(1)} (H_S + d_k).
\] 
Combining the two equations above, we obtain,
\[
\begin{aligned}
\Gamma_{0,4k-3}^{(1)}(t') &=  \int_0^{t'} U_S(t'-t'')
\rho_S^{(0)}(t'') T_k U_S(t'-t'')^\dagger e^{i (t'-t'') d_k} dt''\\
&= \int_0^{t'} \rho_S^{(0)}(t') T_k(t''-t') e^{i (t'-t'') d_k} dt''.
\end{aligned}
\]

Again using the solution properties associated with the superoperator $\exp\left( t\mathcal{L}_0 \right)$, we have that the first block of $\Gamma^{(1)}$ is given by,
\[
\begin{aligned}
&- i \int_0^t U_S(t-t') \Xi^{(1)}(t') U_S(t-t')^\dagger dt'  \\
=& \sum_k \int_0^t U_S(t-t') \big( [\Gamma_{0,4k-3}(t'), T_k] +  [\Gamma_{0,4k-2}(t'), T_k^\dagger]  \big) U_S(t-t')^\dagger dt'.
\end{aligned}
\]
Similar to \cref{eq: rhoS-exp'}, we have also obtained four terms after expanding the commutators. Let us examine the first integral term,
\[
\begin{aligned}
&\sum_k  \int_0^t \int_0^{t'} U_S(t-t') \rho_S^{(0)}(t') T_k(t''-t') T_k e^{i (t'-t'') d_k}  U_S(t-t')^\dagger ) dt''\\
&= \sum_k \int_0^t \int_0^{t'} \rho_S^{(0)}(t') T_k(t''-t) T_k(t'-t) e^{i (t'-t'') d_k} dt'' dt'. 
\end{aligned}
\]
This is the same as the first last integral in \cref{eq: rhoS-exp'}. The rest of the integrals can be similarly verified. 
\end{proof}

\section{The proof of \cref{eq: thm-gamma}}
\begin{proof}
  In \cref{lmm: upert}, we have proved a bound for the case when $\lambda=0.$ Denote the solution by $\Gamma_0(t)$, and the solution operator by $\exp(t\mathcal{L}_0)$. Therefore, the GQME in \cref{eq: qme} can be written in a perturbation form,
\begin{equation}\label{eq: gamma0-gamma}
    \partial_t \Gamma = \mathcal{L}_0 \Gamma -i \lambda (H_1 \Gamma - \Gamma H_1^\dagger). 
\end{equation}

The solution can be recast in an integral form,
\begin{equation}
   \Gamma(t) = \Gamma_0(t) - i \lambda  \int_0^t \exp\left((t-\tau) \mathcal{L}_0\right) (H_1 \Gamma - \Gamma H_1^\dagger) d\tau. 
\end{equation}

From \cref{lmm: p-term},  the super-operator $\exp(t\mathcal{L}_0)$ is bounded. Therefore, we have,
\[ \| \Gamma(t) -  \Gamma_0(t) \| \le   2\lambda C \|H_1\|  \int_0^t \|\Gamma(t) \| dt.\]
As a result, the bound can be obtained by directly using Gronwall's inequality.

Since $ \Gamma_0(t) $ is block diagonal, the above inequality also shows that 
the off-diagonal blocks of $\Gamma(t)$ are of order $\lambda.$ Finally,
the bounds for the trace can be verified from \cref{eq: g00-11}, \cref{eq: g22} and \cref{eq: g33} in  the proof of \cref{lmm: upert}.
\end{proof}

\section{The proof of \cref{eq: tr-bound}}
\begin{proof}
 From \cref{eq: gamma0-gamma}, we may take the trace. 
 \begin{equation}
      \tr(\Gamma(t) ) - \tr(\Gamma^\sz(t) )  = - i \lambda  \int_0^t  \tr( \Sigma(t,\tau) ) d\tau, 
 \end{equation}
 where, 
 \[ \Sigma(t,\tau) =\exp(t\mathcal{L}_0) \Xi,\]
 with $\Xi$ from \cref{eq: Xi}.
By the definition of $\exp(t\mathcal{L}_0)$, $\Sigma(t,\tau)$ is the solution of the equation,
 \[ \partial_t \Sigma = \mathcal{L}_0 \Sigma, \quad  \Sigma(\tau, \tau)=
 (H_1 \Gamma - \Gamma H_1^\dagger).
 \]
 Using the property in \cref{eq: tr-bd} of the super-operator $\exp(t\mathcal{L}_0)$ in \cref{eq: thm-gamma}, we obtain the bound,
 \[
\mid   \tr \big( \Sigma(t,\tau) \big)\mid   \leq 
 3K \mid   \tr\big( \Sigma_{0,0} (\tau, \tau) \big)\mid    + \sum_{k>0} \mid   \tr\big( \Sigma_{k,k} (\tau, \tau) \big)\mid  . 
 \]
 
 We now invoke the estimate in \cref{lmm: p-term}. The trace of each diagonal block of $\Sigma(\tau, \tau)$ is bounded by the off-diagonal blocks of $\Gamma(t)$, which is of order $\lambda$. Namely, there exists a constant, such that,
 \[\mid   \tr \big( \Sigma(\tau,\tau) \big)\mid   \leq  C \lambda.\]
Collecting terms, we have, 
\[
\mid \tr(\Gamma(t) ) - \tr(\Gamma^\sz(t) )\mid   \le Ct  \lambda^2.  
\] 

Alternatively, one can start with \cref{eq: gamma0-gamma}, and apply the formula again to replace the $\Gamma(t)$ in the integral:
\[
\begin{aligned}
     \Gamma(t)  - \Gamma^\sz(t)   =&  - i \lambda  \int_0^t 
 \exp\left((t-\tau) \mathcal{L}_0\right) (H_1 \Gamma^\sz(\tau) - \Gamma^\sz(\tau) H_1^\dagger) d\tau\\  
  & - \lambda^2  \int_0^t 
  \int_0^\tau 
  \exp\left((t-\tau) \mathcal{L}_0\right) 
  (H_1 \Sigma(\tau,s) - \Sigma(\tau,s) H_1^\dagger) ds
   d\tau + \mathcal{O}(\lambda^3).
\end{aligned}
\]
Here $\Sigma(\tau,s) = \exp\left((\tau-s) \mathcal{L}_0\right) \Xi_0$, with
\begin{equation}\label{eq: Xi0}
    \Xi^\sz=H_1 \Gamma_0(s) - \Gamma_0(s) H_1^\dagger.
\end{equation}

For the $\mathcal{O}(\lambda)$ term, we notice that $\Gamma\sz$ is block diagonal, and so in light of  \cref{lmm: p-term}, the matrix $\Xi^\sz$ has zero trace. This implies that,
\[ \tr\big( \Gamma(t) ) - \Gamma^\sz(t) \big) =  \mathcal{O}(\lambda^2).\] 

For the $\mathcal{O}(\lambda^2)$ term, from \cref{eq: xi00'}, we have, the trace of the first block of $ H_1 \Sigma(\tau,s) - \Sigma(\tau,s) H_1^\dagger$ is given by,
\begin{equation}\label{eq: trace1}
    \sum_{k=1}^K  \tr\Big( (\Sigma _{0,4k-3} - \Sigma_{4k-2,0}) T_k \Big)+ \sum_{k=1}^K \tr\Big( T_k^\dagger ( \Sigma_{0,4k-2} - \Sigma_{4k-3,0}) \Big).
\end{equation}

Meanwhile, from the proof of \cref{lmm: upert}, we see that 
the superoperator does not change the off-diagonal blocks in the first row and column. Thus, $\Sigma_{0,j}(\tau, s)= \Sigma_{0,j}(s, s) = \Xi_{0,j}^\sz(s)$.

With direct matrix  multiplications, we find that,
\[ \Xi_{4k-3,0}=  -T_k \Gamma_{4k-3,4k-3} - \Gamma_{0,0} T_k^\dagger, \quad \Xi_{4k-2,0} = T_k \Gamma_{4k-2,4k-2}.  \] 
From the proof of \cref{lmm: upert}, we also have $\Gamma_{4k-3,4k-3}(t)=0.$ Combining these steps, we find that
the trace in \cref{eq: trace1} is zero. 
Therefore, we have
\[ \tr(\rho_S(t)) = 1 + \mathcal{O}(\lambda^3). \] 
\end{proof}

\section{The proof of \cref{lemma:dist-m-k}}
\label{sec:proof-dist-m-k}
\begin{proof}
  We use an intermediate superoperator $\mathcal{I}+\delta\mathcal{K}$.   {Assume the Hilbert space $\mathcal{K}$ is acting on has dimension $N$. Consider} a Hilbert space of arbitrary dimension $N'$. For any operator $Q$   {acting} on $\bbc^{N}\otimes\bbc^{N'}$ with $\norm{Q}_1 = 1$, we have
  \begin{equation}
  \begin{aligned}
   & \norm{(\mathcal{M}_{\delta}\otimes\mathcal{I}_{N'} - (\mathcal{I}_{N} + \delta\mathcal{K})\otimes\mathcal{I}_{N'})(Q)}_1 \\ &= \norm{\sum_{j=0}^m(A_j\otimes I)Q(A_j\otimes I)^{\dag} - (Q + \delta(\mathcal{K}\otimes\mathcal{I}_{N'})(Q)}_1 \\                                                                            &= \norm{\delta^2(H\otimes I)Q(H^{\dag}\otimes I)}_1 \\
           &\leq \norm{H\otimes I}^2\\
                       &\leq (\delta\Lambda)^2.
  \end{aligned}
  \end{equation}
  Now, we have
  \begin{align}
    \label{eq:dist-1}
    \norm{(\mathcal{M}_{\delta} - (\mathcal{I}_{N} + \delta\mathcal{K})}_{\diamond} \leq (\delta\Lambda)^2.
  \end{align}

  To bound the distance between $\mathcal{I}+\delta\mathcal{K}$ and $e^{\delta\mathcal{K}}$, we assume $0 \leq \delta\norm{\mathcal{K}}_{\diamond} \leq 1$. Consider any $X$ such that $\norm{X}_1 \leq 1$, we have
  \begin{align}
    \norm{(e^{\delta\mathcal{K}} - (\mathcal{I}+\delta\mathcal{K}))(X)}_1 &= \norm{\sum_{s=2}^{\infty}\frac{\delta^s}{s!}\mathcal{K}^s(X)}_1 \leq \sum_{s=2}^{\infty}\frac{\delta^s}{s!}\norm{\mathcal{K}^s(X)}_1 \\
                                                                          &\leq \sum_{s=2}^{\infty}\frac{\delta^s}{s!}\norm{\mathcal{K}(X)}_1^s \leq (\delta\norm{\mathcal{K}(X)}_1)^2 \leq (\delta\norm{\mathcal{K}}_1)^2,
  \end{align}
  where the penultimate inequality follows from the fact that $e^z - (1+z) \leq z^2$ when $0 \leq z \leq 1$.
  
  Now, we extend this bound to the diamond norm. Note that, for two Hilbert spaces $\bbc^{N}$ and $\bbc^{N'}$,
  \begin{align}
    (e^{\delta\mathcal{K}} - (\mathcal{I}_{N}+\delta\mathcal{K}))\otimes\mathcal{I}_{N'} = e^{\delta(\mathcal{K}\otimes\mathcal{I}_{N'})} - (\mathcal{I}_{N \times N'} + \delta(\mathcal{K} \otimes \mathcal{I}_{N'})).
  \end{align}
  When $N = N'$, we have that $\norm{\mathcal{K}\otimes\mathcal{I}_{N'}}_1 = \norm{\mathcal{K}}_{\diamond}$. This implies that
  \begin{align}
    \norm{(e^{\delta\mathcal{K}} - (\mathcal{I}_{N'}+\delta\mathcal{K})}_{\diamond} &= \norm{\left(e^{\delta\mathcal{K}} - (\mathcal{I}_{N}+\delta\mathcal{K})\right)\otimes \mathcal{I}_{N'}}_1 \\
                                                                                              &= \norm{e^{\delta(\mathcal{K}\otimes\mathcal{I}_{N'})} - (\mathcal{I}_{N\times N'}+\delta(\mathcal{K}\otimes\mathcal{I}_{N'}))}_1 \\
                                                                                              &\leq (\delta\norm{\mathcal{K}\otimes\mathcal{I}_{N'}}_1)^2 \\
                                                                                              \label{eq:dist-2}
                                                                                              &\leq (\delta\norm{\mathcal{K}}_{\diamond})^2.
  \end{align}

  To see the relationship between $\norm{\mathcal{L}}_{\diamond}$ and $\norm{\mathcal{K}}_1$, first observe that $\norm{\mathcal{K}}_1 \leq 2\Lambda$. Then using the fact the $\norm{M\otimes I} = \norm{M}$ for all $M$, it follows that $\norm{\mathcal{K}}_{\diamond} \leq 2\Lambda$. Together with \cref{eq:dist-2}, we have
  \begin{align}
    \label{eq:dist-3}
    \norm{(e^{\delta\mathcal{K}} - (\mathcal{I}_{N'}+\delta\mathcal{K})}_{\diamond} \leq (2\delta\Lambda)^2.
  \end{align}

  This result in the lemma now directly follows from \cref{eq:dist-1,eq:dist-3}.
\end{proof}

\end{document}